\documentclass[11pt,a4paper]{article}
\usepackage[margin=3cm]{geometry}
\usepackage{amsmath, amsfonts, amssymb, amsthm}
\usepackage[usenames,dvipsnames]{xcolor} 	
\usepackage{graphicx} 	
\usepackage[group-separator={,}]{siunitx}
\usepackage{framed}		
\usepackage{complexity}
\usepackage{arydshln}	
\usepackage{chngcntr}	
\usepackage{algorithm}
\usepackage{algpseudocode}
\usepackage{booktabs}
\usepackage{mathtools} 	
\usepackage{cite}

\usepackage{hyperref}	
\hypersetup{citecolor=violet, linkcolor=red, urlcolor=blue, colorlinks=true, pdfborder={0 0 0}}
\usepackage[capitalise]{cleveref}

\usepackage{thmtools}

\usepackage{bbold}
\newcommand{\one}{\mathbb{1}}
\newcommand{\bra}[1]{\left\langle #1\right\rvert}
\newcommand{\ket}[1]{\left\lvert#1\right\rangle}
\newcommand{\braket}[2]{\left\langle #1\middle|#2\right\rangle}
\newcommand{\ketbra}[2]{\left\lvert #1\middle\rangle\middle\langle#2\right\rvert}
\newcommand{\proj}[1]{\ketbra{#1}{#1}}
\newcommand{\braopket}[3]{\left\langle #1\middle|#2\middle|#3\right\rangle}
\newcommand{\dif}{\,\mathrm{d}}
\newcommand{\abs}[1]{\left|#1\right|}
\newcommand{\norm}[1]{\left\|#1\right\|}
\newcommand{\set}[2]{\left\lbrace #1 \middle| #2 \right\rbrace}

\let\Re\relax
\DeclareMathOperator{\Re}{Re}
\DeclareMathOperator{\Tr}{Tr}
\DeclareMathOperator{\Span}{Span}
\DeclareMathOperator{\Trunc}{Trunc}
\DeclareMathOperator{\ls}{ls} 
\DeclareMathOperator{\cont}{cont} 

\theoremstyle{plain}
\newtheorem{thm}{Theorem} 
\newtheorem{lem}[thm]{Lemma}

\newtheorem{clm}[thm]{Claim} 
\newtheorem{applem}{Lemma}[section]
\theoremstyle{definition}
\newtheorem{defn}[thm]{Definition}

\declaretheoremstyle[notefont=\bfseries,notebraces={}{},headpunct={},postheadspace=1em]{mystyle}
\declaretheorem[style=mystyle,numbered=no,name=Claim]{clmhand} 

\newenvironment{proc}[1]{
\vspace{-.5cm}
\begin{center}
\rule{.95\textwidth}{1pt}
\begin{minipage}{.9\textwidth}
~\\
\textsc{#1}:%
}{
\vspace{.2cm}
\end{minipage}
\rule{.95\textwidth}{1pt}
\end{center}
\vspace{-.3cm}
}

\title{Computing the Degenerate Ground Space of Gapped Spin Chains in Polynomial Time}
\author{Christopher T.\ Chubb~~and~~Steven T.\ Flammia\\
	{\footnotesize\itshape 
		Centre for Engineered Quantum Systems, School of Physics,}\\
	{\footnotesize\itshape
		The University of Sydney, Sydney, NSW 2006, Australia.}}

\begin{document}
\maketitle
\begin{abstract}
Given a gapped Hamiltonian of a spin chain, we give a polynomial-time algorithm for finding the degenerate ground space projector. The output is an orthonormal set of matrix product states that approximate the true ground space projector up to an inverse polynomial error in any Schatten norm, with a runtime exponential in the degeneracy. Our algorithm is an extension of the recent algorithm of Landau, Vazirani, and Vidick~\cite{LandauVaziraniVidick2013} for the nondegenerate case, and it includes the recent improvements due to Huang~\cite{Huang2014'}. The main new idea is to incorporate the local distinguishability of ground states on the half-chain to ensure that the algorithm returns a complete set of global ground states.
\end{abstract}


Strongly correlated quantum systems are at the heart of such diverse physical phenomena as frustrated magnets, high-temperature superconductors, and topological quantum phases, to name but a few, and understanding their low-temperature behaviour is a grand challenge at the interface of physics and computer science. The computational tools to study these systems numerically are however inevitably stymied by the exponential growth of the Hilbert space of an $n$-particle system. 

Many numerical tools have been developed to surmount the computational difficulties surrounding the exponential growth of the state space. Tools like quantum Monte Carlo~\cite{DosSantos2003,FoulkesMitasNeedsRajagopal2001}, coupled cluster methods~\cite{BartlettMusial2007}, and the density-matrix renormalisation group (DMRG)~\cite{Schollwock2011} all exploit some underlying physical structure to perform heuristically efficient simulations of interacting quantum systems. However, like all heuristics they often lack basic theoretical guarantees except in certain special cases, and some of them have failure modes that are difficult to diagnose or characterise. This makes it hard to precisely determine  which physical features enable efficient simulation and which features might be responsible in the cases where the heuristic fails. 

Very recently, a success story has emerged in the rigorous understanding of gapped 1D spin chains. Hastings' proof of the 1D area law~\cite{Hastings2007} established that approximating the ground state of a gapped spin chain is in \NP. The proof technique crucially used the idea of matrix product states (MPS)~\cite{VerstraeteMurgCirac2008}, an ansatz for quantum systems that exploits the area law structure of entanglement~\cite{VerstraeteCirac2004}. The first rigorous nontrivial algorithm for finding ground states, due independently to Schuch and Cirac~\cite{SchuchCirac2010} and Aharonov, Arad, and Irani~\cite{AharonovAradIrani2010}, used dynamic programming to find MPS ground states. The runtime of the algorithm was exponential in the bond dimension, and since Hastings' result used a polynomial bond dimension to achieve a sufficiently accurate approximation, the runtime could only be guaranteed to be exponential for general gapped 1D systems. The first subexponential-time algorithm was found by Arad, Kitaev, Landau, and Vazirani~\cite{AradKitaevLandauVazirani2013}. They used the same algorithm as Refs.~\cite{SchuchCirac2010, AharonovAradIrani2010}, but showed that the ground state could be approximated to sufficient accuracy with an MPS of bond dimension $n^{o(1)}$, where $n$ is the length of the chain. This line of research culminated in the breakthrough algorithm of Landau, Vazirani, and Vidick (LVV)~\cite{LandauVaziraniVidick2013}, who gave the first polynomial-time algorithm for approximating the ground state in the case of a non-degenerate ground space. This was subsequently improved by Huang~\cite{Huang2014'} who gave a deterministic algorithm with an exponentially improved scaling with respect to the gap that also applied to highly frustrated systems. 

In this paper, we extend the LVV algorithm and its subsequent improvements by Huang to the case where the Hamiltonian has a degenerate ground space. To give a precise statement of our results, we first introduce some notation.

Consider a 1D chain of $d$-level quantum systems (\emph{qudits}) with $d = \mathcal{O}(1)$ and a Hamiltonian $H$ that is a sum of local terms on contiguous sites, and let $n$ be the number of sites in the chain. After sufficient coarse-graining, shifting the energy, and dividing by the largest operator norm from the terms in the Hamiltonian, we can take the Hamiltonian to be of the form $H=\sum_{i=1}^{n-1}H_i$, where $0\leq H_i\leq \one$ and $H_i$ acts non-trivially only on the $i$ and $i+1$ qudits. We additionally require that this transformed Hamiltonian has a constant lower bound $\epsilon_1-\epsilon_0 \ge \epsilon \ge \Omega(1)$ on the spectral gap, i.e.\ the difference between the first excited energy ($\epsilon_1$) and the ground state energy ($\epsilon_0$) of $H$ is bounded from below by $\epsilon$ independent of $n$. Our Hamiltonians will have a degenerate spectrum in the ground space with a degeneracy $g = \mathcal{O}(1)$. We say that such a Hamiltonian is in \emph{standard form}, though this is not a unique normal form because of the ambiguity in the coarse-graining step. 

Our main result is an approximation of the ground space in the form of an orthonormal basis of matrix product states that approximate the true ground space projector. More precisely, we have the following.
\begin{thm}[Main result]
	\label{thm:main}
	Given an $n$-qudit Hamiltonian $H$ in standard form and constant bounds $\epsilon$ and $g$ on the gap and degeneracy respectively, then for every sufficiently small approximation error $\eta$ satisfying $1/\eta = n^{\mathcal{O}(1)}$, Algorithm~\ref{alg:overall} constructs the orthonormal approximate ground states $\lbrace\ket{\gamma_j}\rbrace_j$ in matrix product form for $1\leq j\leq g$ in runtime $n^{\mathcal{O}(1)}$. The approximation error gives a bound in the Frobenius norm of
	\[ \bigl\lVert G-\Upsilon \bigr\rVert_F \leq \eta\,, \]
	where $G$ is the true ground space projector, and $\Upsilon :=\sum_{j=1}^{g}\ketbra{\gamma_j}{\gamma_j}$ is the approximate ground space projector.
\end{thm}
The proof of this statement is given in \cref{subsec:main proof}, and a more precise dependence of the runtime on the gap $\epsilon$ and the ground space degeneracy $g$ is made explicit in \cref{subsec:runtime}. Note that since both $G$ and $\Upsilon$ have rank $g$, we can get an approximation error in the trace norm at the cost of a factor of $\sqrt{2g} = \mathcal{O}(1)$ increase in error. 

Our result relies heavily on the ideas in LVV~\cite{LandauVaziraniVidick2013} and the improvements of Huang~\cite{Huang2014'}, so let us recall briefly how this algorithm proceeds. The LVV algorithm proceeds by constructing a \emph{viable set}, which roughly speaking is a small set of matrix product states on a half-chain that is guaranteed to contain the local Schmidt vectors of an approximate ground state on the full chain. In each step of the algorithm, the half-chain of the current viable set is extended by one site to create a new viable set. This initial extension is too large, but trimming away some states with too high an energy yields another suitably small set that is still viable. 

The main new idea is an extension of the size trimming step to degenerate systems, based on the concept of \emph{local distinguishability} of ground states. The algorithm works by iteratively sweeping across the chain, and finds the degenerate ground states one-by-one in a manner similar to that of LVV on each iteration. As we run these iterations, we must ensure that the viable set of potential ground states on the half-chain contains a witness that not only has low energy, but also does not simply reduce to something in the span of the previously found ground states. To ensure this, we break the analysis into two cases of \emph{high-} and \emph{low-}distinguishability. 

In the low-distinguishability case, the Schmidt vectors on the left half-chain of the previously found ground states and the desired witness have large overlap. By recycling these Schmidt vectors, we therefore get a viable set for this new witness for free.

In the high-distinguishability case however, this recycling would not provide a viable set. In this case we need to trim the viable set in a way that not only ensures the existence of a low-energy witness, but a witness which has low overlap with the existing states. If the relevant Schmidt vectors of the desired witness and the existing states have low overlap, then we can optimise for a state of low energy \emph{and} overlap efficiently. 

\section{Preliminaries}
\label{sec:prelim}

Let us begin by giving a more precise statement of the LVV algorithm, with the definitions suitably extended to our treatment of the degenerate case. 

Given the Hamiltonian $H = \sum_{j=1}^{n-1} H_j$, we will often consider a partition of the system into the first $i$ spins and the rest, we will call this the $i|i+1$ cut. The Hamiltonian on a contiguous subset of spins $\{j,\ldots,k\}$ is denoted $H_{[j,k]} = \sum_{i=j}^{k-1} H_i$. In this notation the local terms can also be expressed as $H_i = H_{[i,i+1]}$, and the full Hamiltonian as $H=H_{[1,n]}$. In the special case that an index $i$ is clear from context, we define $H_L = H_{[1,i]}$ and $H_R = H_{[i+1,n]}$. Note that this implies that $H = H_L + H_i + H_R$. The Hilbert space on which $H_{[j,k]}$ acts is denoted by $\mathcal{H}_{[j,k]}$. Similarly, we have the Hilbert space $\mathcal{H} = \mathcal{H}_{[1,n]}$, and similarly for the other possible subscripts. It will sometimes be convenient to subtract a constant energy shift from a Hamiltonian, in which case we define $H_L' = H_L - \epsilon_L$, where $\epsilon_L$ is the least eigenvalue of $H_L$, and similarly for $H_R'$ and $\epsilon_R$. When a particular state $\rho$ is clear from context, we refer to $\Tr(\rho H_L)$ as the left energy of $\rho$, and similarly for the right energy. 

A \emph{viable set} is simply a set of vectors that can be used to construct the Schmidt vectors of a new approximate ground state on a half-chain, assuming we already have approximate ground states $\ket{\gamma_1}$,\ldots,$\ket{\gamma_{h-1}}$. The aim is to inductively grow the region of viability to the entire system starting from one end of the chain. The span of such a set gives the desired small domain on which candidate ground states can be found efficiently. 

\begin{defn}(Viable Set)
	\label{defn:viable} A set of states $S$ is ($h$,$i$,$s$,$b$,$\delta$ or $\Delta$)-viable if:
	\begin{itemize}
		\item The states in $S$ are supported on the first $i\leq n$ qudits of the system, i.e.\ $S\subset \mathcal{H}_{[1,i]}$.
		\item The cardinality is bounded $\abs{S}\leq s$.
		\item Each state in $S$ is represented by an MPS with bond dimension at most $b$.
		\item There exists a witness state $\ket{v}$ which has, depending on context, either
		\begin{itemize}
			\item ($\delta$ case) overlap with the ground space projector $G$ satisfying $\norm{G\ket{v}} \ge 1-\delta$, or
			\item ($\Delta$ case) the expected energy satisfying $\braopket{v}{H}{v} \le \epsilon_0+\Delta\epsilon$.
		\end{itemize}
		\item The reduced density operator of this witness must be supported on $\Span(S)$.
		\item The witnesses state $\ket{v}$ is orthogonal to all previously constructed approximate ground states $\ket{\gamma_1}$,\ldots,$\ket{\gamma_{h-1}}$.
	\end{itemize}
\end{defn}

When it is clear from the context (and it almost always will be clear), we will drop the explicit $h$ dependence in the definition of viable set and refer to a ($i,s,b,\delta$)-viable set, and similarly for $\Delta$.

The two error parameters specified, $\delta$ or $\Delta$, will be used contextually and we will sometimes refer to them as the overlap error and the energy error, respectively. The two error parameters however can be used interchangeably, though at the cost of some overhead which we would like to avoid. The proof is simple and can be found in~\cite{LandauVaziraniVidick2013,Huang2014'}.

\begin{lem}[Interchangeability of $\delta$ and $\Delta$\cite{LandauVaziraniVidick2013,Huang2014'}]
	\label{lem:interchange}
	For the ground state projector $G$ of a Hamiltonian $H$, and any $\ket{v}\in \mathcal{H}$, we have
	\begin{align*}
	\braopket{v}{H}{v}&\leq\epsilon_0+\Delta\epsilon & &\implies \qquad\qquad \|G \ket{v}\|\geq1-\Delta\\
	\|G \ket{v}\|&\geq1-\delta & &\implies \qquad\qquad \braopket{v}{H}{v}\leq \epsilon_0+2\delta n
	\intertext{which allows us to convert between the two measures of viability}
	(\cdot,\cdot,\cdot,\cdot,\Delta)&\text{-viable} & &\implies \qquad\qquad \,(\cdot,\,\cdot,\,\cdot,\,\cdot,\,\delta\,=\,\Delta\,)\text{-viable}\\
	(\cdot,\cdot,\cdot,\cdot,\,\delta)&\text{-viable} & &\implies \qquad\qquad (\cdot,\cdot,\cdot,\cdot,\Delta=\frac{2n}{\epsilon}\delta)\text{-viable}\,.
	\end{align*}
\end{lem}

For later use, we define the following absolute constants $c$ and $c'$ that are independent of all other parameters, and a simple function $f$ of the ground state degeneracy,
\begin{align*}
c:=&6.25\times 10^{-54} & c':=&10^{-11} & f:=&g(2g+1) \,,
\end{align*}
where we note that $200\sqrt[4]{c}= c'$. These values can be used to make all of the constants hidden by the big-O notation in our derivations and procedures completely explicit. However, our final result still relies on certain implicit estimates (in particular, the estimates in Refs.~\cite{Osborne2006,Huang2014}) so we are not able to give a completely explicit scaling.


\section{The Algorithm}

\begin{algorithm}[ht!]
	\caption{Ground space approximation}\label{alg:overall}
	\begin{algorithmic}[1]
		\Procedure {DegenerateGSA}{$H$, $g$, $\epsilon$, $\eta$}
		\State $\ket{\gamma_1}\leftarrow$\textsc{ NondegenerateGSA}($H$, $\epsilon$, $\eta$)\Comment{Run non-degenerate algorithm}
		\State $E\leftarrow\braopket{\gamma_1}{H}{\gamma_1}$\Comment{Approximate the ground energy}
		\ForAll{$h\in [2, g]$}
		\State $S_{h,0}\leftarrow \lbrace 1\rbrace$ \Comment{Initial viable set}
		\ForAll{$i\in [1, n-1]$}
		\State $E_i\leftarrow \braopket{\gamma_1}{H_{[1,i]}}{\gamma_1}$ \Comment{Estimate energy of half-chain}
		\State $L_{h,i}\leftarrow$\textsc{ SchmidtVecs}($i$, $\ket{\gamma_{1}}$,$\ldots$,$\ket{\gamma_{h-1}}$) \Comment{Extract left Schmidt vectors}
		\State $S_{h,i}\leftarrow$\textsc{ Step}($S_{h,i-1}$, $L_{h,i}$, $H$, $\epsilon$, $E$, $E_i$)\Comment{Step viable set one qudit along}
		\EndFor
		\State $S_{h,n}\leftarrow$\textsc{ FinalStep}($S_{h,n-1}$, $H$, $\epsilon$, $E$, $\eta$)\Comment{Low-error final step}
		\State $\ket{\gamma_h}\leftarrow$\textsc{ ApproxGroundState}($H$, $S_{h,n}$, $\ket{\gamma_1}$, $\ldots\,$, $\ket{\gamma_{h-1}}$)\Comment{Extract witness}
		\EndFor
		\State \Return $\ket{\gamma_1},\,\ldots\,,\,\ket{\gamma_g}$
		\EndProcedure
	\end{algorithmic}
\end{algorithm}

Algorithm \ref{alg:overall} works by first running the slight modifications on the existing nondegenerate ground state approximation algorithms of Ref.~\cite{LandauVaziraniVidick2013,Huang2014'}, which is discussed in more detail in \cref{subsubsec:nondegen}.
As discussed earlier, the extended algorithm works by iteratively finding ground states one-by-one, with each guaranteed to be orthogonal to those previously found until all $g$ ground states have been found.
It is not necessary to know \emph{a priori} a specific value for $g$, but unless $g = \mathcal{O}(1)$ holds, then there is no guarantee from the area law alone that all the ground states can be approximated well by the algorithm.
Moreover, if the algorithm is iterated more times than the true degeneracy of the system, then it will necessarily produce high energy states which will herald this over-estimation, as we will show in \cref{app:overlap}. As such we will henceforth assume that the degeneracy $g$ is known.

The \textsc{Step}/\textsc{FinalStep} procedures will be discussed below in \cref{subsec:step}, and its four sub-procedures will follow in Sections~\ref{subsec:extension}-\ref{subsec:reduction}. Finally the \textsc{ApproxGroundState} procedure will be discussed in \cref{subsec:gsa}.

\subsection{Step}
\label{subsec:step}

\begin{algorithm}[th!]
	\caption{Stepping functions}\label{alg:step}
	\begin{algorithmic}[1]
		\Procedure {Step}{$S_{h,i-1}$, $L_{h,i}$, $H$, $\epsilon$, $E$, $E_i$}
		\State $S_{h,i}^{(1)}\leftarrow$\textsc{ Extend}($S_{h,i-1}$) \Comment{Increment $i$}
		\State $S_{h,i}^{(2)}\leftarrow$\textsc{ Trim}($S_{h,i}^{(1)}$, $L_{h,i}$, $H$, $\epsilon$, $E_i$)\Comment{Reduce cardinality of set}
		\State $S_{h,i}^{(3)}\leftarrow$\textsc{ Truncate}($S_{h,i}^{(2)}$, $\epsilon$)\Comment{Truncate state bond dimensions}
		\State $S_{h,i}^{(4)}\leftarrow$\textsc{ Reduce}($S_{h,i}^{(3)}$, $H$, $\epsilon$, $E$)\Comment{Push back down the error level}
		\State \Return $S_{h,i}^{(4)}$
		\EndProcedure
		\Statex
		\Procedure {FinalStep}{$S_{h,n-1}$, $H$, $\epsilon$, $E$, $\eta$}
		\State $S_{h,n}^{(1)}\leftarrow$\textsc{ Extend}($S_{h,n-1}$) \Comment{Extend to final qudit}
		\State $S_{h,n}^{(2)}\leftarrow$\textsc{ Trim}($S_{h,n}^{(1)}$, $L_{h,n}$, $H$, $\epsilon$, $E_n$)
		\State $S_{h,n}^{(3)}\leftarrow$\textsc{ Truncate}($S_{h,n}^{(2)}$, $\epsilon$)
		\State $S_{h,n}^{(4)}\leftarrow$\textsc{ FinalReduce}($S_{h,n}^{(3)}$, $H$, $\epsilon$, $E$, $\eta$)\Comment{Reduce error to desired level}
		\State \Return $S_{h,n}^{(4)}$
		\EndProcedure
	\end{algorithmic}
\end{algorithm}

\begin{table}[tbh!]
	\centering
	\begin{tabular}{cccccc} \toprule
		Set & $i$ & $s$ & $B$ & $\delta$ & $\Delta$ \\ \midrule
		$S_{0}$ & $0$ & $1$ & $1$ & $0$ & $0$\\
		$S_{1}$ & $1$ & $s$ & $b$ & {\color{gray}$c\epsilon^6/f^4$} & {$c\epsilon^6/f^4$}\\
		\vdots & \vdots & \vdots & \vdots & \vdots & \vdots\\
		$S_{i-1}$ & $i-1$ & $s$ & $b$ & $\color{gray}c\epsilon^6/f^4$ & $c\epsilon^6/f^4$\\
		{} & {} & {} & {} & {} & {}\\\hdashline
		{} & {} & {} & {} & {} & {}\\
		$S_{i}\left\lbrace
		\begin{matrix}
		S_{i}^{(1)}\\
		S_{i}^{(2)}\\
		S_{i}^{(3)}\\
		S_{i}^{(4)}			\end{matrix}\right.$ &
		\renewcommand\arraystretch{1.25} $\begin{matrix}
		{\color{red}i}\\
		i\\
		i\\
		i
		\end{matrix}$ &
		\renewcommand\arraystretch{1.25} 
		$\begin{matrix}
		ds\\
		{\color{red}p_1}\\
		p_1+q\\
		pp_1+pq+q
		\end{matrix}$ &
		\renewcommand\arraystretch{1.25} 
		$\begin{matrix}
		b\\
		dsb+q^2\\{
			\color{red}p_2}\\
		pp_2
		\end{matrix}$ & 
		\renewcommand\arraystretch{1.25} 
		$\begin{matrix}
		{\color{gray}c\epsilon^6/f^4}\\
		1/100\\
		1/5\\
		{\color{gray}c\epsilon^6/f^4}
		\end{matrix}$ & 
		\renewcommand\arraystretch{1.2} 
		$\begin{matrix}
		c\epsilon^6/f^4\\
		{\color{gray} n/50\epsilon}\\
		{\color{gray}2n/5\epsilon}\\
		{\color{red}c\epsilon^6/f^4}
		\end{matrix}$  \\
		{} & {} & {} & {} & {} & {}\\\hdashline
		{} & {} & {} & {} & {} & {}\\
		$S_{i+1}$ & $i+1$ & $s$ & $b$ & $\color{gray}c\epsilon^6/f^4$ & $c\epsilon^6/f^4$\\
		\vdots & \vdots & \vdots & \vdots & \vdots & \vdots\\
		$S_{n-1}$ & $n-1$ & $s$ & $b$ & $\color{gray}c\epsilon^6/f^4$ & $c\epsilon^6/f^4$\\
		$S_{n}$ & $n$ & $p_1+2q$ & $p_0 p_2+q$ & $\color{gray}\eta^2/4f$ & $\eta^2/4f$\\\bottomrule
	\end{tabular}
	\caption{The viability parameters of the sets produced during Lines 8--14 of \cref{alg:overall}. Red indicates the desired parameter change in each sub-step, grey the weaker of the two error parameters as given by \cref{lem:interchange}. Here $p$ is some absolute polynomial $p=n^{\mathcal{O}(1)}$ and the remaining polynomials scale as $p_0,p_1,p_2,q=n^{a}$ where $a=\mathcal{O}\bigl(1+\epsilon^{-1}\sqrt{\log \eta^{-1}/\log n}\bigr)$. The polynomials $s:=pp_1+pq+q$ and $b:=pp_2$ denote the fixed points of the cardinality and bond dimension under the \textsc{Step} procedure.}
	\label{tab:params}
\end{table}

The \textsc{Step} procedure, shown above in Algorithm~\ref{alg:step}, takes a viable set defined on $i-1$ qudits, and constructs a new set which is viable on $i$ qudits. As well as physically extending the set, this procedure can be efficiently performed such that neither the cardinality nor the error grow. Starting from the trivially ($0$,$1$,$1$,$0$)-viable set $\lbrace 1\rbrace$, the \textsc{Step} procedure allows you to inductively construct ($i$,$s$,$b$,$\Delta=c\epsilon^6/f^4$)-viable sets $S_{h,i}$, for all $i$. Here the parameters $s$ and $b$ are all polynomial functions of $n$ and are given in \cref{tab:params}. On the final qudit the \textsc{Step} procedure is tweaked in the form of \textsc{FinalStep}, which performs a more powerful error reduction procedure to bring the error down to the desired final level. 

By constructing a viable set on the entire system, a polynomial-sized subspace has been constructed on which the Hamiltonian may be efficiently optimised, allowing the previously inaccessible witness state to be harvested, yielding the final desired ground state approximation. This final optimisation is performed in \textsc{ApproxGroundState}, which follows the \textsc{Step}-based viable set construction.

\subsection{Extension}
\label{subsec:extension}
By induction, suppose we start with a ($i-1$,$s$,$b$,$\Delta=c\epsilon^6/f^4$)-viable set $S_{i-1}$. The first step is to take a viable set on $i-1$ qudits and extend it trivially to a viable set on $i$ qudits. To do this we take the element-wise tensor product of the existing viable set with a basis on the $i$th qudit. It is trivial to see that the set resulting from \textsc{Extension}, $S_{h,i}^{(1)}$, is ($i$,$ds$,$b$,$\Delta=c\epsilon^6/f^4$)-viable; see Ref.~\cite{LandauVaziraniVidick2013} for the explicit argument. We simply note that the net result of this step is to increase the cardinality of the viable set from $s$ to $ds$, while keeping all other parameters fixed. 

\subsection{Size trimming}
\label{subsec:trimming}
For the size-trimming we are going to use two nets to compensate for our ignorance of two properties of the ground states: the local energy differences between them, and their entanglement structure across the $i|i+1$ cut. The former is represented by the difference in the expectation value of $H_L$ between our desired witness and (say) the first ground state found. The latter is represented by the boundary contraction (defined below) of the witness state in question. In both cases we are going to construct a net on the space of possible values of these unknowns, and perform the necessary optimisations in the neighbourhoods of these elements, thus allowing for error guarantees in spite of this ignorance. 

\subsubsection{Boundary contraction}
\label{subsubsec:boundary contraction}

Viable sets, like reduced density matrices, capture local information about a given state. As defined, viable sets are solely determined by the left Schmidt vectors of a witness state; the reduced density operator however also contains the Schmidt coefficients across the cut. As such, merely finding vectors of low energy on the half-chain $[1,i]$ in no way guarantees consistency with a low-energy global state, as they might not be consistent with the entanglement structure across the $i|i+1$ cut. As well as this the system may be frustrated, with a globally low-energy state appearing locally excited. Both of these problems are circumvented by performing the desired optimisations alongside \emph{boundary contraction} constraints.

Given a state of Schmidt rank $B$ and Schmidt decomposition across the $i|i+1$ cut given by $\ket{v}=\sum_{j=1}^{B}\lambda_j\ket{a_j}\ket{b_j}$, let $U_v:\mathbb{C}^B\to \mathcal{H}_R$ be the partial isometry specified by $U_v\ket{j}=\ket{b_j}$. Then we have the following.
\begin{defn}[Boundary Contraction~\cite{LandauVaziraniVidick2013}]
	\label{defn:bc}
	The \emph{left state} of $\ket{v}$ with respect to the $i|i+1$ cut is defined as 
	\[
	\ket{\ls(v)}:=U_v^\dagger\ket{v}=\sum_j\lambda_j\ket{a_j}\ket{j}\in\mathcal{H}_L \otimes\mathbb{C}^B\,.
	\]
	The \emph{boundary contraction} of $\ket{v}$ is defined as
	\[ \cont(v):=\Tr_{[1,i-1]}\left(\ketbra{
		\ls(v)}{\ls(v)}\right)=U_v^\dagger\Tr_{[1,i-1]}\left(\ketbra{v}{v}\right)U_v \,.
	\]
\end{defn}
Note that the density matrix $\cont(v)$ is supported on $\mathcal{H}_i\otimes\mathbb{C}^B$.

The interpretation of a boundary contraction is the following. The difference in right energies between a candidate state and a valid witness state for the full chain is bounded by how close the boundary contractions of the states are. By extending a state by the as-yet-unknown witness state, we can argue the existence of a low energy solution within our extended viable set.

\begin{lem}[\cite{LandauVaziraniVidick2013,Huang2014'}]
	\label{lem:boundary contraction energy}
	Let $\sigma$ be a density matrix on $\mathcal{H}_L\otimes \mathbb{C}^B$ and $\ket{v}=\sum_{j=1}^{B}\lambda_j\ket{a_j}\ket{b_j}\in\mathcal{H}$ be a Schmidt-decomposed state. The density matrix $\sigma':=U_v\sigma U_v^\dagger$ on $\mathcal{H}$ satisfies
	\begin{align*}
	\Tr\left(\sigma'H\right)&\leq \Tr(\sigma H_L)+\braopket{v}{\left(H_i+H_R\right)}{v}
	\\&\quad
	+\norm{\Tr_{\left[1,i-1\right]}\left(\sigma\right)-\cont(v)}_1\left(1+\norm{H_R\Pi_v}-\braopket{v}{H_R}{v}\right)\,.
	\end{align*}
	where $\Pi_v = U_vU_v^\dagger$ is the rank $B$ projector from the isometry $U_v$. 
\end{lem}

\begin{proof}
	We give a sketch of the proof and refer the reader to Ref.~\cite{LandauVaziraniVidick2013} for more detail. Begin with $\Tr(\sigma' H) = \Tr\bigl[(\sigma'+\proj{v}-\proj{v})H\bigr]$ for some ground state $\ket v$. Expanding $H = H_L + H_i + H_R$ and using the matrix H\"older inequality together with the bound $\lVert H_i \rVert \le 1$ yields the proof. 
\end{proof}

\subsubsection{Nets}
\label{subsubsec:net}
Recall that an \emph{$\epsilon$-net} on a metric space is a set of points such that any point in the space is within distance $\epsilon$ of a point in the net. We will use two nets in our argument: one on the interval $[-1,1+\eta]$ in the standard absolute value metric and one on the set of matrices with bounded max-norm on $\mathbb{C}^d\otimes\mathbb{C}^B$ in the trace norm.

We denote by $\mathcal{M}_{\eta}$ the $\eta/2$-net on $[-1,1+\eta]$ whose points are 
\[ \mathcal{M}_{\eta}=\left\lbrace-1,-1+\eta,\cdots,1-\eta,1,1+\eta\right\rbrace\,. \]
The cardinality of $\mathcal{M}_{\eta}$ is at most $\left\lceil2/\eta\right\rceil+1$, and for any real value $x\in\left[-1-\eta,1+\eta\right]$ there exists a $y\in\mathcal{M}_{\eta}$ such that the following \emph{one-sided} error bound holds, $0\leq y-x\leq\eta$. We refer to this net as the \emph{energy net}. The choice of limits $[-1,1+\eta]$ is justified in \cref{app:frust}.

Let $\mathcal{N}_\eta$ be the matrices on $\mathbb{C}^d\otimes\mathbb{C}^B$ with entries whose real and imaginary parts are in the set 
\[ \left\lbrace-1,-1+\eta/B^2d^2,\cdots,1-\eta/B^2d^2,1\right\rbrace\, \]
The cardinality of $\mathcal{N}_\eta$ is upper bounded by $\left(2\left\lceil Bd/\eta\right\rceil+1\right)^{4Bd}$. For any positive semidefinite matrix $\rho$ on $\mathbb{C}^d\otimes\mathbb{C}^B$ with trace at most $1$, there exists a point $X\in\mathcal{N}_\eta$ such that $\norm{\rho-X}_1\leq \eta$. We refer to this net as the \emph{boundary contraction net}.

\subsubsection{Procedure}

Let $L_{h,i}$ be the set of left Schmidt vectors with respect to the $i\vert i+1$ cut of the states 
$\ket{\gamma_1},\dots,\ket{\gamma_h}$. Furthermore define $\Pi_{h,i}$ to be the projector 
onto $\Span(L_{h,i})$. Now define parameters $\delta = 8\sqrt{c}\epsilon^3/f^2$, $\eta = 4c'\epsilon/f$, and $\xi$ to be chosen later in Eq.~\eqref{eqn:xi}. 

\begin{proc}{Trim}
	For each boundary contraction $X\in\mathcal{N}_{\xi/2}$ with 
	bond dimension $B_\delta$ (to be defined below in Eq.~\eqref{eqn:Bdelta}) and left-energy difference $Y\in\mathcal{M}_{\eta}$, take $\sigma$ to be a density matrix supported on 
	$\Span\Bigl(S_{h,i}^{(1)}\cup L_{h,i}\Bigr)\otimes \mathbb{C}^B$. 
	We will specifically choose $\sigma$ to be a solution to the \emph{degenerate size-trimming convex program}:
	\begin{align}
	\label{prog:trim}
	\text{min}~~~& \Tr(\Pi_{h,i}\sigma)\\
	\text{where}~~~& \Tr\bigl[H_L(\sigma-\ketbra{\gamma_1}{\gamma_1})\bigr] \leq Y,\notag\\
	&\norm{\Tr_{[1,i-1]}(\sigma)-X}_1\leq\xi/2,\notag\\
	&\Tr(\sigma)=1,~~\sigma\geq 0\notag.
	\end{align}
	Now consider the set of vectors $\lbrace\ket{v_k}\rbrace_k$ consisting of all the eigenvectors of $\sigma$ with eigenvalues at least $10^{-9}/g$. Now define
	\begin{align}
	S_{h,i}^{(2)} = \bigl\lbrace\ket{v_{k,j}}\bigr\rbrace_{j,k} \cup L_{h,i}\,,
	\end{align}
	where $\ket{v_{k,j}}$ is the $j$th left Schmidt vector of $\ket{v_k}$ across the $i|B$ cut.
\end{proc}

The next Claim, whose proof is postponed until the next section, establishes that the viability parameters after each \textsc{Trim} step remain bounded as advertised. 

\begin{clm}[{\hyperref[clm:step2proof]{Proof}} in \cref{subsec:trimming analysis}]
	\label{clm:step2}
	The set resulting from \textsc{Trim}, $S_{h,i}^{(2)}$, is $(i,p_1,dsb+q^2,\delta=1/100)$-viable, where 
	\[ p_1(n),\,q(n)=n^{a}\text{ where }a=\mathcal{O}\bigl(1+\epsilon^{-1}\sqrt{\log \eta^{-1}/\log n}\bigr)\,. \]
\end{clm}

\subsection{Bond truncation}
\label{subsec:truncation}

First, we recall some results on low-rank approximate ground states~\cite{VerstraeteCirac2006,AradKitaevLandauVazirani2013,Huang2014}. Let us define a function
\begin{align}
\label{eqn:Bdelta}
B_\delta:= 2^{\mathcal{O}(\epsilon^{-1}+\epsilon^{-1/4}\log^{3/4}\delta^{-1})}\,.
\end{align}
Specifically we will take $B_{\delta}$ to be defined by the following lemma, a corollary to the area law~\cite{Hastings2007,AradKitaevLandauVazirani2013,Huang2014}.
\begin{lem}[\!\cite{AradKitaevLandauVazirani2013,Huang2014}]
	\label{lem:area law}
	For any ground state $\ket{\Gamma}$ and bipartite cut of the system in question, there exists a state $\ket{\gamma}$ with Schmidt rank across that cut bounded by $B_\delta$ such that $\abs{\braket{\Gamma}{\gamma}}\geq 1-\delta$. Moreover, there exists an MPS approximation $\ket{\gamma'}$ with bond dimension across \emph{every} cut of at most $B_{\delta/n}$ and for which $\abs{\braket{\Gamma}{\gamma'}}\geq 1-\delta$.
\end{lem}

If we know the Schmidt decomposition of a given state across a bipartite cut, then the Eckart-Young theorem~\cite{EckartYoung1936} guarantees that the \emph{truncation} of that state to Schmidt-rank $D$ is the best approximation with said Schmidt rank. 

\begin{defn}[Truncation]
	\label{defn:truncation}
	Consider a state $\ket{v}$ with Schmidt decomposition $\ket{v}=\sum_j\lambda_j\ket{a_j}\ket{b_j}$ across the $i | i+1$ cut, where any degeneracy in the Schmidt coefficients is resolved arbitrarily, and with non-increasing $\lambda_j$. Given an integer $D$, define the \emph{truncation} of $\ket{v}$ by
	\begin{align*}
	\Trunc_D^i\ket{v}:=\sum_{j=1}^D\lambda_j\ket{a_j}\ket{b_j}\,.
	\end{align*}
\end{defn}
This approximation is optimal in the following sense~\cite{EckartYoung1936}. For any state $\ket{v}$, and any integer $D$, the truncated vector $\ket{v'}:=\Trunc_D\ket{v}/\norm{\Trunc_D\ket{v}}$ satisfies $\braket{v}{v'}\geq\abs{\braket{v}{w}}$ for any $\ket{w}$ of Schmidt rank at most $D$. 

Combining these results, we get a precise statement that low-energy states can also be approximated faithfully with low-rank approximations. See Ref.~\cite{Huang2014'} for a proof.

\begin{lem}[\!\cite{Huang2014'}]
	\label{lem:truncation}
	For a state $\ket{v}$ with $\braopket{v}{H}{v}\leq\epsilon_0+\Delta\epsilon$ where $\Delta(1+\epsilon)\leq1/4$, the truncated state $\ket{w}:=\Trunc^i_{B_{\Delta\epsilon}}\ket{v}/\norm{{\Trunc^i_{B_{\Delta\epsilon}}\ket{v}}}$ has energy $\braopket{w}{H}{w}\leq\epsilon_0+48\sqrt{\Delta}$.
\end{lem}

Building on these lemmas, we can define the bond truncation procedure \textsc{Truncate} for our algorithm as follows. 
\begin{proc}{Truncate}
	Take each vector within $S_{h,i}^{(2)}$, and truncate all bonds on $1,\dots,i-1$ to $P:=800nB_{1/800n}=n^{1+o(1)}$. Take $S_{h,i}^{(3)}$ to be 
	\begin{align*}
	S_{h,i}^{(3)}:=\set{\Trunc_{P}^1\Trunc_{P}^2\cdots\Trunc_{P}^{i-1}\ket{s}}{\ket{s}\in S_{h,i}^{(2)}}\cup L_{h,i}\,.
	\end{align*}
	\vspace{-.5cm}
\end{proc}

As before, the following Claim establishes that the \textsc{Truncate} step preserves the viability with the advertised parameters. 

\begin{clm}[{\hyperref[clm:step3proof]{Proof}} in \cref{subsec:truncation analysis}]
	\label{clm:step3}
	The set resulting from \textsc{Truncate}, $S_{h,i}^{(3)}$, is $(i,p_1+q,p_2,\delta=1/5)$-viable, where 
	\[ p_2(n)=n^a \quad \text{where } a ={\mathcal{O}\bigl(1+\epsilon^{-1}\sqrt{\log \eta^{-1}/\log n}\bigr)}\,. \]
\end{clm}

\subsection{Error reduction}
\label{subsec:reduction}
Next step in the algorithm is to push the error level back down. To do this we need to construct operators called approximate ground state projectors (AGSPs). One convenient choice of AGSP would be of the form
\[ A:=\exp\biggl[-\frac{x(H-\epsilon_0)^2}{2\epsilon^2}\biggr] \,,\]
where $x$ is a parameter defined later. We cannot however construct such an operator. Firstly it would require us to exactly know the ground state energy. Secondly it would require exponential time to construct and specify.

We already have a ground state approximation $\ket{\gamma_1}$ which we got from the non-degenerate algorithm with energy at most $\epsilon_0+\eta^2\epsilon/4f$. As such, if we define $\epsilon_0':=\braopket{\gamma_1}{H}{\gamma_1}$ then for sufficiently small error $\eta$ we have
\[ \abs{\epsilon_0-\epsilon_0'}\leq \eta^2\epsilon/4f \le \epsilon/2\,. \]

In order to construct an approximation of $A$ by a more efficient operator $K$, we define
\begin{align}
\label{eq:K_AGSP}
K:=\frac{2\epsilon\tau}{\sqrt{2\pi x}}\sum_{j=0}^{\lceil T/\tau\rceil}\exp\left[i\epsilon_0'\tau j-\frac{\epsilon^2\tau^2j^2}{2x}\right]\,U_B(\tau j) \,,
\end{align}
where $U_B(t)$ is an approximation of $\exp(-iHt)$ with Schmidt rank $B$ across every cut, the construction and analysis of which is given in Ref.~\cite{Osborne2006}.

The approximation between $A$ and $K$ is broken into three distinct approximations. The first is that the integral on $(-\infty,\infty)$ is truncated to an integral on $[-T,T]$, with an exponentially small error in $T$. The second is that this integral is approximated by a Riemann sum, specifically the rectangle rule\footnote{Higher-order approximations such as other Newton-Cotes formulae could be used, though this won't change the overall scaling of the various parameters.} with a discrete step-size of $\tau$. The third is that $\exp(-iHt)$ is approximated by a unitary $U_B(t)$ that has low Schmidt rank. All three errors are considered in \cref{subsec:reduction analysis}. If we take parameters in Eq.~\ref{eq:K_AGSP} to have the scaling
\begin{align*}
x:=&\mathcal{O}(\log\zeta^{-1})\\
T:=&\mathcal{O}(\epsilon^{-1}\log(1/\zeta)\sqrt{\log (n/\zeta)})\\
\tau^{-1}:=&\mathcal{O}(n^2\zeta^{-1}\sqrt{\log n/\zeta})\\
B:=&(\zeta^{-1})^{\mathcal{O}(1/\epsilon)}\cdot\poly(n/\zeta)\,,
\end{align*}
then this approximate AGSP is capable of lowering the energy error down to $\Delta=\zeta$. As such we will henceforth refer to this as a $\zeta$-approximate AGSP. With the definition of $K$ in hand, we can give a precise specification of \textsc{Reduce}:

\begin{proc}{Reduce}
	Decompose a $\zeta$-approximate AGSP with $\zeta = c\epsilon^6/f^4$ as $K=\sum_jA_j\otimes B_j$. Then return the set given by applying this approximate AGSP, combined with the recycled Schmidt vectors of the previous approximate ground states.
	\[S_{h,i}^{(4)}:=\set{A_j\ket{s}}{\forall j,\,\ket{s}\in S_{h,i}^{(3)}}\cup L_{h,i}\,.\]
	\vspace{-.5cm}
\end{proc}

As we will prove in \cref{subsec:reduction analysis}, the \textsc{Reduce} step can lower the error of a viable set from $\delta=1/5$ down to $\Delta=\zeta$, increasing both the bond dimension and cardinality by a factor of 
\[(1/\zeta)^{\mathcal{O}(1/\epsilon)}\cdot(n/\zeta)^{\mathcal{O}(1)}\,.\]
Therefore the above construction with $\zeta=c\epsilon^6/f^4$ can send a
($i$,$p_1+q$,$p_2$,$\delta=1/5$)-viable set to a ($i$,$pp_1+pq+q$,$pp_2$,$c\epsilon^6/f^4$)-viable set where 
\[p(n)=n^{\mathcal{O}(1)}\,.\]
This forms our next Claim.

\begin{clm}[{\hyperref[clm:step4proof]{Proof}} in \cref{subsec:reduction analysis}]
	\label{clm:step4}
	The set resulting from \textsc{Reduce}, $S_{h,i}^{(4)}$, is $(i,pp_1+pq+q,pp_2,\Delta=c\epsilon^6/f^4)$-viable.
\end{clm}

\subsubsection{Final error reduction}

As for the final step, we now want to reduce the error further than was previously necessary. We do this by simply constructing a stronger version of the previously used approximate AGSP.

\begin{proc}{FinalReduce}
	Construct a $\eta^2/4f$-approximate AGSP $K$. Return the set constructed by applying this AGSP, combined with the previous approximate ground states,
	\[S_{h,n}^{(4)}:=\set{
		K\ket{s}}{ \ket{s}\in S_{h,n}^{(3)}}\cup\lbrace\ket{\gamma_1},\cdots, \ket{\gamma_{h-1}}\rbrace\,.\]\vspace{-.5cm}
\end{proc}

\begin{clm}[{\hyperref[clm:step4proof]{Proof}} in \cref{subsec:reduction analysis}]
	\label{clm:step4final}
	The set resulting from \textsc{FinalReduce}, $S_{h,n}^{(4)}$, is $(n,p_1+q+g,p_0p_2,\Delta=\eta^2/4f)$-viable, where 
	\[ p_0(n)=n^a , \quad \text{where } a ={\mathcal{O}\bigl(1+\epsilon^{-1}\sqrt{\log \eta^{-1}/\log n}\bigr)}\,. \]
\end{clm}

\subsection{Ground state approximation}
\label{subsec:gsa}

Once we have $S_{h,n}$, a viable set on the entire system, we then need to extract the witness state. By construction, a viable set provides a domain on which known optimisation techniques can be applied to efficiently find approximate ground states. We summarise this with the last procedure. 

\begin{proc}{ApproxGroundState}
	Take $\sigma$ to be a density matrix supported on $\Span(S_{h,n})$, found as the solution to the convex program
	\begin{align}
	\text{min}\quad&\Tr(H\sigma_h)\label{prog:gsa}\\
	\text{where}\quad&
	\braopket{\gamma_j}{\sigma_h}{\gamma_j}=0\text{ for }1\leq j<h\,,\notag\\
	&\Tr(\sigma_h)=1\,,\,\,\sigma_h\geq 0\,.\notag
	\end{align}
	Return $\ket{\gamma_h}$, a leading eigenvector of $\sigma_h$.
\end{proc}

The objective functions and constraints can both be evaluated efficiently, firstly because our witness states are MPS with small bond dimension, and the Hamiltonian is a matrix product operator (MPO) with small bond dimension, and also because the domain of the optimisation is polynomially large. Since the program takes the form of a convex optimisation, these two properties are sufficient to allow the program to be solved efficiently. 
By optimality of the above program, the energy of $\sigma_h$ must be at most that of the witness, $\epsilon_0+\eta^2\epsilon/4f$. If we take $\eta \leq 1/3$ and using \cref{lem:demixing} we get that the leading eigenvector, $\ket{\gamma_h}$, has an energy at most $\epsilon_0+\eta^2\epsilon/4g$, and ground space overlap $\lVert G\ket{\gamma_h}\rVert \ge 1-\eta^2/4g$.


\subsection[Proof of Theorem 1]{Proof of \cref{thm:main}}
\label{subsec:main proof}
We can now turn to our proof of the main result \cref{thm:main}. The proof is conditional on Claims~\ref{clm:step2}, \ref{clm:step3}, \ref{clm:step4}, and \ref{clm:step4final}, whose proof is given below in \cref{sec:analysis}.
\begin{proof}
	The Claims~\ref{clm:step2}, \ref{clm:step3}, \ref{clm:step4}, and \ref{clm:step4final} establish that in each iteration of Algorithm~\ref{alg:overall} a viable set with energy error at most $\eta^2/4f$ is produced, and \cref{lem:demixing} establishes that this yields an orthonormal set of vectors such that $\bra{\gamma_j}G\ket{\gamma_j} \ge 1-\eta^2/2g$ for all $j=1,\ldots,g$. Define $\Upsilon :=\sum_{j=1}^{g}\proj{\gamma_j}$, and compute
	\begin{align*}
	\lVert G-\Upsilon\rVert_F^2 &= \Tr(G-\Upsilon)^2 = 2g - 2\Tr(G\Upsilon) \le 2g - 2 g (1-\eta^2/2g) = \eta^2 \,.
	\end{align*}
	This completes the proof of the main result conditional on the Claims.
\end{proof}

\subsection{Runtime}
\label{subsec:runtime}

As all the parameters given in \cref{tab:params} are polynomials in $n$, and all operations required to execute \cref{alg:overall} (convex optimisation, inner product of MPS etc.) can be performed at polynomial overhead, the overall runtime is also polynomial in $n$. The leading order term in $n$ gives a runtime of $T=n^{a}$ where $a=\mathcal{O}\bigl(1+\sqrt{\log \eta^{-1}/\log n}\bigr)$, which reduces to
\[ T=\begin{dcases}
n^{\mathcal{O}(1/\epsilon)} & \text{where }\eta=n^{-\mathcal{O}(1)} \,,\\
n^{\mathcal{O}(1)} & \text{where }\eta=n^{-o(1)} \,.
\end{dcases} \]

In Ref.~\cite{Huang2014'''}, an improved analysis of the \textsc{FinalReduce} step (which we will discuss in \cref{subsubsec:AAGSP}) is presented. The scaling  the polynomials $p_0$, $p_1$, and $p_2$ are all improved to $q=n^{\mathcal{O}(1)}$ in the $\eta^{-1}=n^{\mathcal{O}(1)}$ case. This means that the corresponding run-time reduces to $T=n^{\mathcal{O}(1)}$.

One thing to note is that this leading order scaling is independent of the degeneracy. If we now focus on the degeneracy scaling, fixing $n$ and allowing $g$ to vary, we find that the dominant term now comes from the cardinality of the boundary contraction net. As such the leading-order scaling becomes $2^b$ where $b=g^{\tilde{\mathcal{O}}(1/\epsilon)}$; this is also the leading order scaling in the inverse-gap.


\section{Analysis}
\label{sec:analysis}

\subsection{Preliminaries}

To prove the claims in the previous section, we first state some simple results that are easy to prove. These results were also used in Refs.~\cite{LandauVaziraniVidick2013,Huang2014'}, and we state them here without proof. 

\begin{lem}[Overlap results~\cite{LandauVaziraniVidick2013,Huang2014'}]
	\label{lem:omnibus}
	~
	\begin{itemize}
		\item For any states for which $\abs{\braket{v}{w}}\geq 1-\delta$ and $\abs{\braket{v'}{w}}\geq 1-\delta'$, then $\abs{\braket{v}{v'}}\geq 1-2(\delta+\delta')$.
		\item For any two states such that $\abs{\braket{v}{w}}\geq 1-\delta$, $\abs{\braopket{v}{O}{v}-\braopket{w}{O}{w}}\leq 2\sqrt{2\delta}$ for any operator $O$ with $\norm{O}\leq 1$.
		\item If $\abs{\braket{u}{\Gamma}}\leq \omega$ and $\abs{\braket{v}{u}}\geq 1-\delta$ then $\abs{\braket{v}{\Gamma}}\leq \omega+\sqrt{2\delta}$.
		\item For a pair of vectors with bounded overlap $\abs{\braket{u_1}{u_2}}\leq \omega$, and a second pair $\ket{v_1}$, $\ket{v_2}$ such that $\abs{\braket{u_j}{v_j}}\geq1-\delta$ for $j=1,2$, then we have $\abs{\braket{v_1}{v_2}}\leq \omega+\sqrt{10\delta}$.
	\end{itemize}
\end{lem}

\begin{lem}[Orthogonalising lemma]
	\label{lem:ortho}
	Suppose you have two vectors $\ket{u}$ and $\ket{v}$, such that $\braopket{u}{H}{u},\braopket{v}{H}{v}\leq \epsilon_0+\Delta\epsilon$, with bounded overlap $\abs{\braket{u}{v}}\leq \beta$. Then there exist orthogonal vectors $\ket{u'},\ket{v'}\in \Span\lbrace\ket{u},\ket{v}\rbrace$, such that $\braopket{u'}{H}{u'},\braopket{v'}{H}{v'}\leq \epsilon_0+\Delta \epsilon (1+\beta)/(1-\beta)$.
\end{lem}
\begin{proof}
	The proof is constructive, take the two vectors to be
	\[ \ket{u'}=\ket{u}\,,\qquad\ket{v'}=\frac{\ket{v}-\ket{u}\braket{u}{v}}{\norm{\ket{v}-\ket{u}\braket{u}{v}}}\,. \]
	As $\ket{v'}\propto(I-\ketbra{u}{u})\ket{v}$ this vector is clearly orthogonal to $\ket{u'}$. As for the energy bound, define $H':=H-\epsilon_0$, then
	\begin{align*}
	\braopket{v'}{H'}{v'}
	&=\frac{\braopket{v}{H'}{v}-2\Re\braopket{v}{H'}{u}\braket{u}{v}+\braket{v}{u}\braopket{u}{H'}{u}\braket{u}{v}}{\norm{\ket{v}-\ket{u}\braket{u}{v}}^2}\\
	&\leq \frac{1+\abs{\braket{u}{v}}}{1-\abs{\braket{u}{v}}}\Delta\epsilon\leq \frac{1+\beta}{1-\beta}\Delta\epsilon
	\end{align*}
	After reintroducing the $\epsilon_0$ term in $H$ we get the final energy bound. 
\end{proof}

	If we apply the same argument to the projector orthogonal to the ground space, the same $(1+\beta)/(1-\beta)$ growth can also be shown to occur in the overlap error. The Schmidt rank of the resulting orthogonalised vectors will in general be the sum of the two original vectors' Schmidt ranks. If we can bound the overlap between a vector and \emph{any} member of a vector space, then that vector can be similarly orthogonalised with respect to that whole space at the same cost.
	
	As expectation values are bilinear in pure states, the optimisations in the \textsc{Trim} and \textsc{ApproxGroundState} are both performed over mixed states such that their objective functions are linear. To aid in the analysis of these optimisations, we will next show that the leading eigenvector of a low energy mixed state must be a low energy pure state.
	
\begin{lem}[Demixing]
	\label{lem:demixing}
	If a mixed state $\sigma$ has an energy at most $\epsilon_0+\Delta\epsilon/(2g+1)$ where $\Delta\leq1/3g$, then the leading eigenvector has energy at most $\epsilon_0+\Delta\epsilon$.
\end{lem}
\begin{proof}
	First decompose the mixed state as $\sigma=\sum_{k> 0}\lambda_k\ketbra{v_k}{v_k}$, where $\lambda_k\geq 0$. By the normalisation of $\sigma_k$ we have $\sum\lambda_k=1$. Next define the set of indices $K$ as
	\[ K=\set{k\,}{\,\braopket{v_k}{H}{v_k}\leq\epsilon_0+\Delta\epsilon}\,. \]
	By comparing the energy bound of the mixed state as a whole as well as those on indices in $K$, we get
	\begin{align*}
	\epsilon_0+\Delta\epsilon/(2g+1)
	&\geq\Tr(\sigma H)\\
	&=\sum_{k\in K}\lambda_k\braopket{v_k}{H}{v_k}+\sum_{k\notin K}\lambda_k\braopket{v_k}{H}{v_k}\\
	&\geq\epsilon_0+\Delta\epsilon\sum_{k\notin K}\lambda_k
	\end{align*}
	which in turn gives $\sum_{k\in K}\lambda_k\geq 2g/(2g+1)$, analogous to a Markov inequality. As we required $\Delta\le1/3g$, each vector corresponding to a index in $K$ has an energy at most $\epsilon_0+\epsilon/3g$, and thus a ground space overlap at least $1-1/3g$. By \cref{applem:fullness} there can only be $g$ such orthonormal vectors, leading us to conclude that $\abs{K}\leq g$. Using this we can bound the largest eigenvalue corresponding to an index in $K$
	\begin{align*}
	\max_{k\in K}\lambda_k
	&\geq\underset{k\in K}{\mathrm{mean }}\lambda_k=\frac{\sum_{k\in K} \lambda_k}{\abs{K}}\geq \frac{2}{2g+1}
	\intertext{as well as that \emph{not} in $K$}
	\max_{k\notin K}\lambda_k
	&\leq \sum_{k\notin K}\lambda_k\leq \frac{1}{2g+1}\,.
	\end{align*}
	As such we can conclude that the largest eigenvalue must belong within $K$, meaning that the leading eigenvector has energy at most $\epsilon_0+\Delta\epsilon$.
\end{proof}

\subsubsection{Non-degenerate ground state approximation}
\label{subsubsec:nondegen}

Here we will briefly address \textsc{NondegenerateGSA}, which we have claimed is a small modification on the nondegenerate ground state approximation algorithms of Ref.~\cite{LandauVaziraniVidick2013,Huang2014'}. The analysis of these algorithms largely holds in the degenerate case, appropriately replacing the notions of ground \emph{state} overlap with ground \emph{space} overlap. The only key difference is the use of a special case of \cref{lem:demixing}. As \cref{lem:demixing} will be invoked during the size-trimming and ground state approximation steps (Sections~\ref{subsec:gsa},\,\ref{subsec:trimming analysis}), the energy levels in these steps grow by a factor of $g$ overlooked in the non-degenerate analysis. By applying a slightly more powerful error reduction, dropping error levels  beforehand by $g$ to compensate, the same overall scaling can be proved, at the mere introduction of an extra constant factor. Due to this rather superficial change we will use these algorithms without proof of their validity. Furthermore we will take 
\begin{align}
q(n):=n^{a}\quad\text{where}\quad a=\mathcal{O}(1+\epsilon^{-1}\sqrt{\log \eta^{-1}/\log n})\label{eqn:q}
\end{align}
to denote the bond dimension of this first approximate ground state $\ket{\gamma_1}$ given by this algorithm. We will show that the bond dimension of all subsequent ground state approximations follows the same scaling, and thus use $q(n)$ to generally denote the bond dimension of any combination thereof. 

\subsection{Size trimming}
\label{subsec:trimming analysis}

Next we will consider the size trimming sub-step, and prove the associated \cref{clm:step2}. In \cref{lem:lowenergy} we will show that there exists a low-energy state orthogonal to the existing ground states, which we hope to find. After this we will introduce and define local-distinguishability, and give some intuition about how our size-trimming works in the low- and high-distinguishability regimes. We will then consider how the error levels behave in these two regimes, showing that the size-trimming procedure yields a good viable set for any level of distinguishability. We will conclude this section in with a full proof of \cref{clm:step2}.

\subsubsection{Low-energy state}

The optimisation \cref{prog:trim} in \textsc{Trim} is performed with respect to $H_L$ and not the whole Hamiltonian. There is no reason to assume \emph{a priori} that states which have low energy with respect to the global Hamiltonian $H$, have low energy with respect to half-chain Hamiltonians $H_L$ or $H_R$. Using the locality of the Hamiltonian however we will see that a global ground state of $H$ can be excited by an energy at most $1$ with respect to $H_L+H_R$.

\begin{lem}
	For any ground state $\ket{\Gamma}$ of $H$,
	\[ \bra{\Gamma}\left(H_L'+H_R'\right)\ket{\Gamma}\leq1 \]
\end{lem}
\begin{proof}
	Consider the state $\ket{v}:=\ket{l}\otimes\ket{r}$, where $\ket{l}$ and $\ket{r}$ are in the $0$-eigenspace of $H_L'$ and $H_R'$ respectively, which are non-empty by construction. The energy of $\ket{v}$ is 
	\[ \braopket{v}{H}{v}=\epsilon_L+\epsilon_R+\braopket{v}{H_i}{v}\,, \]
	so for $\ket{v}$ to not have an energy lower than that of a ground state, we require $\epsilon_0\leq\epsilon_L+\epsilon_R+\braopket{v}{H_i}{v}$. By bounding this further by the normalisation $\norm{H_i}\leq 1$, this means the ground expectation of $H_L'+H_R'$ is at most $1$.
\end{proof}

As global ground states have low energy with respect to $H_L+H_R$, there must exist a good approximant of the ground state in the low-energy subspace of $H_L+H_R$. As such we will restrict our analysis to this subspace, which we will refer to as the truncated Hilbert space.

\begin{defn}[Truncated Hilbert Spaces]
	Let $P_t$ be the projection onto $\mathcal{H}_{[1,i]}\otimes\mathcal{H}_{[i+1,n]}^{\leq t}$, where $\mathcal{H}_{[i+1,n]}^{\leq t}$ is the subspace of $\mathcal{H}_{[i+1,n]}$ spanned by eigenstates of $H_R'$ with eigenvalues at most $t$, that is,
	\[ 
	\mathrm{im}(P_t) = \mathrm{span}\set{ \ket{v}}{H'_R \ket{v} = \lambda \ket{v} \, ; \ \lambda \le t} \,.
	\]
	Similarly let $Q_t$ be the projection onto $\left(\mathcal{H}_{[1,i]}\otimes \mathcal{H}_{[i+1,n]}\right)^{\leq t}$ which corresponds to the span of the eigenvectors of $H_L'+H_R'$ with eigenvalues at most $t$, that is, 
	\[
	\mathrm{im}(Q_t) = \mathrm{span}\set{ \ket{v}}{(H'_R + H'_L)\ket{v} = \lambda \ket{v} \, ; \ \lambda \le t} \,.
	\]
As the left and right Hamiltonians commute and are positive, we have $P_t\geq Q_t$.
\end{defn}

The following choices for $t$ and $\xi$ ensure that the error of the boundary contraction step remains small. 
Define $t:=99\left(\log99+4\log f-\log c-6\log\epsilon\right)$ such that $99\cdot 2^{-t/99}=c\epsilon^6/f^4$. Furthermore without loss of generality we can take $\epsilon\leq10^9\leq\sqrt[6]{70f^4/c}$ such that \[t\geq 1+4\sqrt{c}\epsilon^3/f^2\geq \braopket{\Gamma}{(H_L'+H_R')}{\Gamma}+4\sqrt{c}\epsilon^3/f^2\] for any ground state $\ket{\Gamma}$. Finally we define the spacing of the boundary contraction net $\xi$
\begin{align}
\xi:=c'\epsilon/f(1+t)=\Omega((\epsilon/g)/\log(g/\epsilon))\,.\label{eqn:xi}
\end{align}

In general whilst a ground state can be heavily frustrated, containing contributions from states with arbitrarily high local energy, these contributions will be small. Specifically we expect that a given ground state will have a large overlap with $P_t$ and $Q_t$ --- this result is known as the \emph{truncation lemma}, and was proven in Refs.~\cite{AradKitaevLandauVazirani2013,Huang2014'}.

\begin{lem}[Truncation lemma~\cite{AradKitaevLandauVazirani2013,Huang2014'}]
	\label{lem:trunc}
	For any ground state $\ket{\Gamma}$ of $H$,
	\[\norm{(1-P_t)\ket{\Gamma}}\leq \norm{(1-Q_t)\ket{\Gamma}}\leq 99\cdot2^{-t/99}\,. \]
\end{lem}

Given these results, we can now move on to proving the existence of a low-energy witness state which \cref{prog:trim} will approximate.

\begin{lem}
	\label{lem:lowenergy}
	There exists a state $\ket{w}$ in $\Span\lbrace S_i^{(1)}\rbrace\otimes\mathcal{H}_{[i+1,n]}^{\leq t}$ orthogonal to the previously constructed approximate ground states $\ket{\gamma_1},\,\dots,\,\ket{\gamma_{h-1}}$, of Schmidt rank $B:=B_{\delta}$ where $\delta=8\sqrt{c}\epsilon^3/f^2$, with an energy at most $\epsilon_0+c'\epsilon/f$.
\end{lem}
\begin{proof}
	Let $\ket{v'}$ be a witness of $S_i^{(1)}$, such that it is orthogonal to the previous ground state approximations and has energy at most $\epsilon_0+c\epsilon^6/f^4$. By \cref{lem:interchange} there exists a ground state $\ket{\Gamma}$ such that $\abs{\langle v'|\Gamma\rangle}\geq 1-c\epsilon^6/f^4$. Using \cref{lem:omnibus} we get that 
	\[\braopket{v'}{H_i}{v'}\geq\braopket{\Gamma}{H_i}{\Gamma}-2\sqrt{2c}\epsilon^3/f^2\,.\]
	Using this together with the energy of $\ket{v'}$, and again assuming that $c \epsilon^6 / f^4$ is sufficiently small gives
	\begin{align*}
	\bra{v'}\left(H_L'+H_R'\right)\ket{v'}
	&\leq\bra{\Gamma}\left(H_L'+H_R'\right)\ket{\Gamma}+2\sqrt{2c}\epsilon^3/f^2+c\epsilon^6/f^4\\
	&\leq\braopket{\Gamma}{(H_L'+H_R')}{\Gamma}+4\sqrt{c}\epsilon^3/f^2\,. 
	\end{align*}
	
	Decomposing the above using the projector $P_t$ and its orthogonal complement -- both of which commute with $H_L'$ and $H_R'$ -- then this gives
	\begin{align*}
	\braopket{\Gamma}{(H_L'+H_R')}{\Gamma}+4\sqrt{c}\epsilon^3/f^2
	&\geq\braopket{v'}{(H_L'+H_R')}{v'}\\
	&\geq\braopket{v'}{P_t(H_L'+H_R')P_t}{v'}\\
	&\quad+\braopket{v'}{(1-P_t)(H_L'+H_R')(1-P_t)}{v'}\\
	&\geq \braopket{v}{(H_L'+H_R')}{v}\norm{P_t\ket{v'}}^2+t\norm{(1-P_t)\ket{v'}}^2
	\end{align*}
	where we have defined $\ket{v}:=P_t\ket{v'}/\norm{P_t\ket{v'}}$, such that $\ket{v}$ now lies in the desired domain $\Span\lbrace S_i^{(1)}\rbrace\otimes\mathcal{H}_{[i+1,n]}^{\leq t}$ by construction. Using $\norm{P_t|v'\rangle}^2+\norm{(1-P_t)|v'\rangle}^2=1$ as well as the inequality $t\geq \braopket{\Gamma}{(H_L'+H_R')}{\Gamma}+4\sqrt{c}\epsilon^3/f^2$, 
	we can reduce the above simply to 
	\begin{align}
	\braopket{v}{(H_L'+H_R')}{v}\leq\braopket{\Gamma}{(H_L'+H_R')}{\Gamma}+4\sqrt{c}\epsilon^3/f^2
	\,.
	\label{eqn:expectedge}
	\end{align}
	The overlap between this truncated witness and the ground state is bounded
	\begin{align*}
	\abs{\braket{v}{\Gamma}}
	&\geq \abs{\braopket{v'}{P_t}{\Gamma}}\\
	&\geq \abs{\braket{v'}{\Gamma}}-\abs{\braopket{v'}{(1-P_t)}{\Gamma}}\\
	&\geq 1-c\epsilon^6/f^4-99\cdot 2^{-t/99}\\
	&= 1-2c\epsilon^6/f^4\,,
	\end{align*}
	and so once again applying \cref{lem:omnibus} gives 
	\begin{align}
	\braopket{v}{H_i}{v}\leq\braopket{\Gamma}{H_i}{\Gamma}+4\sqrt{c}\epsilon^3/f^2 
	\,.\label{eqn:expectmiddle}
	\end{align}
	Combining Equations \eqref{eqn:expectedge} and \eqref{eqn:expectmiddle} we can bound the total energy of the witness projected into the truncated Hilbert Space
	\[\braopket{v}{H}{v}\leq \epsilon_0+8\sqrt{c}\epsilon^3/f^2\,.\]
	
	Next we need to considering trimming the bond dimension. By applying \cref{lem:truncation} and recalling the definition of $B_\delta$ from \cref{eqn:Bdelta}, we get that the truncated state 
\[\ket{w}:=\Trunc^i_{B}\ket{v}/\norm{\Trunc^i_{B}\ket{v}}\]
has $\braopket{w}{H}{w}\leq\epsilon_0+96\sqrt[4]{4c}\epsilon/f\leq \epsilon_0+0.75c'\epsilon/f\,.$
	
	Next we consider the overlap induced with the existing ground states. To bound this we consider the overlap between $\ket{w}$ and $\ket{v'}$. We had that $\abs{\langle v'|\Gamma\rangle}\geq 1-c\epsilon^6/f^4$ and $\abs{\braket{v}{\Gamma}}\geq 1-2c\epsilon^6/f^4$, \cref{lem:omnibus} therefore gives \[\abs{\langle v'|v\rangle}\geq 1-6c\epsilon^6/f^4\,.\] Next take $\ket{\gamma}:=\Trunc^i_B\ket{\Gamma}/\norm{\Trunc^i_{B}\ket{\Gamma}}$, which has $\abs{\braket{\gamma}{\Gamma}}\geq 1-8\sqrt{c}\epsilon^3/f^2$ by \cref{lem:area law}. \cref{lem:omnibus} once more gives that $\abs{\braket{v}{\gamma}}\geq1-20\sqrt{c}\epsilon^3/f^2$. As $\ket{w}$ is a truncation of $\ket{v}$, by the Eckart-Young theorem it must have a larger overlap than $\ket{\gamma}$
	\[ \abs{\braket{v}{w}}\geq\abs{\braket{v}{\gamma}}-20\sqrt{c}\epsilon^3/f^2 \,. \]
	Applying \cref{lem:omnibus} a third time we get finally that 
	\[ \abs{\braket{v'}{w}}\geq 1-52\sqrt{c}\epsilon^3/f^2\,. \]
	As $\ket{\gamma_1},\dots,\ket{\gamma_{h-1}}$ are all perpendicular to $\ket{v'}$, the overlap between $\ket{w}$ and any element of $\Span\lbrace\ket{\gamma_1},\dots,\ket{\gamma_{h-1}}\rbrace$ is upper bounded by the component of $\ket{w}$ orthogonal to $\ket{v'}$
	\[ \sqrt{1-\abs{\braket{v'}{w}}^2}\leq \sqrt{1-\left(1-52\sqrt{c}\epsilon^3/f^2\right)^2}\leq \sqrt{104}\sqrt[4]{c\epsilon^6}/f\,. \]
	
	Taking $\epsilon\leq 10^{7}$ without loss of generality this overlap is less than $1/7$, so we can orthogonalise $\ket{w}$ to the existing ground states at a multiplicative energy cost of $4/3$, resulting in a final state with $\braopket{w}{H}{w}\leq \epsilon_0+c'\epsilon/f$, and all other desired properties by construction.
	
\end{proof}

\subsubsection{Distinguishability}

Viable sets capture local information about a witness state, specifically the Schmidt vectors. If two states are globally orthogonal, this can mean that their Schmidt vectors span orthogonal spaces, but they can also span identical regions. Indeed global orthogonality of witness states cannot be locally imposed for this reason. 

When considering viable sets, a natural notion of distinguishability is the overlap of the desired witness state, and the Schmidt vectors of the already constructed states. If we take $\ket{w}$ to be a witness state which we wish to capture, $L_{h,i}$ to be the set of Schmidt vectors on $[1,i]$ of the $h$ already constructed ground states, and $\Pi_{h,i}$ the projector onto $\Span(L_{h,i})$, then the distinguishability is defined as
\begin{align*}
D_{h,i}&:=1-\braopket{w}{\Pi_{h,i}}{w}.
\end{align*}
From here we will omit the subscripts on $L$, $\Pi$, and $D$, since they will be clear from context. If we think of $\Pi$ as being the projector onto the vectors on $[1,i]$ which we already have, then $D$ can be considered a measure of how much of the state is unaccounted for.  

Clearly if $D$ is low, then a viable set can be constructed entirely by recycling existing Schmidt vectors. The precise bounds for the quality of viable set constructed in this manner, as a function of $D$, are given in \cref{lem:lowD}.

In the case of a high $D$ new vectors need to be found. In this case we know there is little overlap between the Schmidt vectors of our new witness and $\Span(L)$. Our procedure works by restricting to a low-energy subspace, and projecting away from the existing vectors, by minimising the expectation value of $\Pi$. In \cref{lem:ortho energy} we will show that this procedure must produce a low energy state. Given that our witness is highly distinguishable, \cref{lem:highdist} will show that this new witness after size trimming must also be highly distinguishable. This will culminate in \cref{lem:trim overlap} where we show that this therefore implies that the new witness has small overlap with the previous ground states, producing a good viable set.

We will conclude this subsection by showing in \cref{lem:combined} that by performing both the low- and high-distinguishability procedures and combining the results, we can upper bound the error of our viable set in a distinguishability-independent manner.

\subsubsection[Low D]{Low $D$}

By definition, in the low $D$ case the left Schmidt vectors of $\ket{w}$ have large overlap with $L$. As a result we can see that $L$ will form a viable set, whose witness is the projection of $\ket{w}$ onto $L$, and that the error of this viable set is low for low $D$.

\begin{lem}[Low $D$ Solution]
	\label{lem:lowD}
	There exists a state $\ket{v}$ orthogonal to $\ket{\gamma_1},\dots,\ket{\gamma_{h-1}}$ with $\Tr_{[i+1,n]}\ketbra{v}{v}$ supported on the span of $L_{h,i}$, such that there exists a ground state $\ket{\Gamma}$ with $\abs{\braket{\Gamma}{v}}\geq 1-\delta_{\rm Low}$ where
	\begin{align}
	\delta_{\rm Low}= 1-\frac{1-c'/f-\sqrt{D}}{\sqrt{1-D}}\,.\label{eqn:lowd}
	\end{align}
\end{lem}

\begin{proof}
	The idea here is to project $\ket{w}$ into $\Span\left(L_{h,i}\right)\otimes \mathcal{H}_{[i+1,n]}$ and use the distinguishability to bound the error associated with renormalisation. Specifically take 
	\[\ket{v}:=\frac{\Pi\ket{w}}{\norm{\Pi\ket{w}}}\]
	where we recall that $\braopket{w}{\Pi}{w}=1-D$. As the existing ground state approximations are all stable under $\Pi$ by construction, the orthogonality of $\ket{w}$ to these vectors carries over to $\ket{v}$. The energy of $\ket{w}$ is at most $\epsilon_0+c'\epsilon/f$, by \cref{lem:interchange} this means there is a ground state $\ket{\Gamma}$ with \[\abs{\braket{w}{\Gamma}}\geq 1-c'/f\,.\]
	Decomposing this with the projector $\Pi$, and using both the triangle inequality and Cauchy's ineqality, we get
	\begin{align*}
	1-c'/f
	&\leq \abs{\braket{\Gamma}{w}}\\
	&\leq \abs{\braopket{\Gamma}{\Pi}{w}}+\abs{\braopket{\Gamma}{(I-\Pi)}{w}}\\
	&\leq\sqrt{1-D}\cdot\abs{\braket{\Gamma}{v}}+\sqrt{D}\,.
	\end{align*}
	Rearranging this gives the desired ground state overlap stated above.
\end{proof}

\subsubsection[High D]{High $D$}

Take $\sigma$ to both be the solution to Program~\ref{prog:trim} for the $X$ closest to $\cont(w)$ and the smallest $Y$ which is greater than or equal to $\braopket{w}{H_L}{w}-\braopket{\gamma_1}{H_L}{\gamma_1}$, i.e.\ the tightest set of constraints satisfied by $\ketbra{w}{w}$. By the nature of the two nets, described in \cref{subsubsec:net}, we have that
\[ \norm{X-\cont(w)}_1\leq \xi\qquad\text{and}\qquad0\leq Y-\braopket{w}{H_L}{w}+\braopket{\gamma_1}{H_L}{\gamma_1}\leq 4c'\epsilon/f\,.  \]

Here we want to argue that if $D$ is not too low, then one of the $\sigma$ outputs of \cref{prog:trim} has an eigenvector which can be extended to a good witness. To extend the eigenvectors of $\sigma$ we take the `boundary-uncontracted' version $\sigma'=U_w \sigma U_w^\dagger$, eigendecomposing these matrices as:
\begin{align*}
\sigma&=\sum_k\lambda_k \ketbra{v_k}{v_k}  &  \sigma'&=\sum_k\lambda_k \ketbra{v_k'}{v_k'}
\end{align*}
where $\ket{v'_k}=U_w\ket{v_k}$.

In \textsc{Trim} we kept all eigenvectors whose value exceeded $10^{-9}/g$.
We next want to show that there can only exist a finite number of such eigenvectors, and that they represent the vast majority of the spectrum of $\sigma$.

\begin{lem}
	\label{lem:ortho energy}
	There exist at most $g$ eigenvectors of $\sigma'$ with energy at most $\epsilon_0+\num{30000}c'\epsilon/g$, and the sum of the corresponding eigenvalues is at least $\Lambda = 1-1/(\num{10000}g+\num{5000})$.
\end{lem}
\begin{proof}
	Defining $\sigma'=U_w\sigma U_w^\dagger$, the boundary contraction property \cref{lem:boundary contraction energy} gives an energy bound on $\sigma'$.
	\begin{align*}
	\Tr(H\sigma')
	&\leq \Tr(H_L\sigma)+\braopket{w}{(H_i+H_R)}{w}
	\\&\qquad
	+\norm{\Tr_{[1,i-1]}(\sigma)-\cont(w)}_1
	\left(1+\norm{U_w^\dag H_RU_w}-\braopket{w}{H_R}{w}\right)\\
	&\leq \Tr(H_L\sigma)+\braopket{w}{(H_i+H_R)}{w}
	+\norm{\Tr_{[1,i-1]}(\sigma)-\cont(w)}_1
	\left(1+\norm{U_w^\dagger H_R'U_w}\right)\\
	&\leq\braopket{w}{H}{w}+\bigl(Y-\braopket{w}{H_L}{w}+\braopket{\gamma_1}{H_L}{\gamma_1}\bigr)+\xi(1+t)\\
	&\leq \epsilon_0+6c'\epsilon/f\,.
	\end{align*}
	By applying an argument similar to \cref{lem:demixing}, if we increase the excitation not only by a factor of $(2g+1)$ but also by $5000$, we get that there are between $1$ and $g$ eigenvectors which have energy at most $\epsilon_0+\num{30000}c'\epsilon/g$. Moreover we also get that the sum of the corresponding eigenvalues is at least $1-1/(\num{10000}g+\num{5000}) = \Lambda$.
\end{proof}
Next we rearrange the order of the eigenvectors. Let $1,\dots,g$ refer to these $g$ vectors\footnote{Padding out these indices with zeros as necessary if there are less than $g$ such vectors.} described in \cref{lem:ortho energy}. In other words:
\begin{align*}
\braopket{v_k'}{H}{v_k'}&\leq \epsilon_0+\num{30000}c'\epsilon/g\quad\text{   for }k \leq g\\
\braopket{v_k'}{H}{v_k'}&\geq \epsilon_0+\num{30000}c'\epsilon/g\quad\text{   for }k> g
\end{align*}
Next we drop any such vectors with eigenvalues less than the $10^{-9}/g$ threshold\footnote{Once again padding out with zero vectors such that there are still $g$ such vectors for convenience.}. It should first be noted that by applying \cref{lem:demixing} to $\sigma$ we get that the dominant eigenvalue that is at least $2/(2g+1)$ has a corresponding eigenvector in this set, and so this set is non-empty. Next we note that as there are at most $g-1$ such vectors dropped, so decrease in the sum of eigenvalues is at most $10^{-9}\cdot\frac{g-1}{g}\leq 10^{-9}$. As such we can conclude that the sum of the corresponding eigenvalues is still bounded
\[ \sum_{k=1}^{g}\lambda_k\geq \Lambda-10^{-9}\,. \]

Without loss of generality, take $\ket{v_1'}$ to be the member of $\lbrace\ket{v_1'},\dots,\ket{v_g'}\rbrace$ for which the expectation value of $\Pi$ is the lowest. We now want to show that the expectation of $\Pi$ with respect to $\ket{v_1'}$ not significantly higher than it is for $\ket{w}$ (which is $1-D$ by definition). As such this will imply that in the high $D$ case $\ket{v_1'}$ has small overlap with $L$, meaning that we have succeeded in finding a state with low overlap with the existing ground states.

\begin{lem}
	\label{lem:highdist}
	\[ \braopket{v_1'}{\Pi}{v_1'}\leq \left[\frac{1}{\Lambda}+10^{-9}\right](1-D) \]
\end{lem} 
\begin{proof}
	The specific iteration of \cref{prog:gsa} we are referring to is that for which $\ket{w}$ satisfies all of the constraints of the optimisation. As such, by the optimality of the solution $\sigma$, we know that $\braopket{w}{\Pi}{w}\geq \Tr(\sigma \Pi)$. Decomposing this and using the definition of $D$:
	\begin{align*}
	\braopket{w}{\Pi}{w}&\geq \Tr(\sigma\Pi)\\
	1-D &\geq \sum_k \lambda_k\braopket{v_k}{\Pi}{v_k}\,.
	\intertext{As $\Pi$ acts strictly to the left of the cut and $U_w$ strictly to the right they comute, meaning the expectation value of $\Pi$ for primed and unprimed eigenvectors is identical.}
	1-D&\geq \sum_k\lambda_k\braopket{v_k'}{\Pi}{v_k'}\\
	&\geq \sum_{k=1}^{g}\lambda_k\braopket{v_k'}{\Pi}{v_k'}\,.\\
	\intertext{As we defined $\ket{v_1'}$ to have minimal expectation this can be further loosened to}
	1-D&\geq \left(\sum_{k=1}^{g}\lambda_k\right)\braopket{v_1'}{\Pi}{v_1'}\,.
	\end{align*}
	Rearranging this together with the bound for $\sum_{k=1}^{g}\lambda_k$, we get 
	\[ \braopket{v_1'}{\Pi}{v_1'}\leq \bigl[\Lambda-10^{-9}\bigr]^{-1}(1-D)\,. \]
	After a binomial expansion, followed by a loosening of the inequality, we arrive at our final bound.
\end{proof}

In general $\ket{v_1'}$ will have some non-trivial overlap with the existing ground state approximations. This low expectation value for $\Pi$ can however be used to bound this overlap when $D$ is not too small.

\begin{lem}
	\label{lem:trim overlap}
	For any vector $\ket{\gamma}\in\Span\lbrace\ket{\gamma_1},\dots,\ket{\gamma_{h-1}}\rbrace$ the overlap between $\ket{\gamma}$ and $\ket{v_1'}$ is upper bounded
	\[ \abs{\braket{\gamma}{v_1'}}^2\leq \left[\frac{1}{\Lambda}+10^{-9}\right](1-D)
	\,. \]
\end{lem}
\begin{proof}
	The main property to use here is that $\Pi\ket{\gamma}=\ket{\gamma}$ by definition. Using this and Cauchy-Schwartz gives
	\begin{align*}
	\abs{\braket{\gamma}{v_1'}}^2 =\abs{\braopket{\gamma}{\Pi}{v_1'}}^2 \leq\norm{\Pi\ket{v_1'}}^2 =\braopket{v_1'}{\Pi}{v_1'} \leq \left[\frac{1}{\Lambda}+10^{-9}\right](1-D)\,.
	\end{align*}
\end{proof}

Take $\ket{u}$ to be $\ket{v_1'}$ projected orthogonally to $\ket{\gamma_1},\dots,\ket{\gamma_{h-1}}$, as per \cref{lem:ortho}. Then the ground space overlap of this vector can be bounded to be $\abs{\braket{\Gamma}{u}}\geq 1-\delta_{\rm High}$ for some ground state $\ket{\Gamma}$, where 
\begin{align}
\delta_{\text{High}}=\frac{\num{30000}c'}{g}\frac{1+\beta}{1-\beta}\quad\text{     and}\quad\beta^2\leq\left[\frac{1}{\Lambda}+10^{-9}\right](1-D)\,.\label{eqn:highd}
\end{align}

\subsubsection[All D]{All $D$}

Next we can combine the two above errors. We have shown that there exist bounds on the overlap error $\delta$ in the low- and high-$D$ cases. By combining these, we can bound $\delta$ independent of the unknown $D$.

\begin{lem}
	\label{lem:combined}
	The combined error $\delta=\min (\delta_{\rm Low},\delta_{\rm High})$ over the total range of distinguishabilities is bounded $\delta\leq 0.01$ for all $g\geq 2$ and $D\in\left[0,1\right]$.
\end{lem}
\begin{proof}
	Consider the case of $D=10^{-4}$ and $g=2$. The two error levels in this case can be explicitly calculated, and are 
	\begin{align*}
	\delta_{\text{Low}}(g=2,\,D=10^{-4})
	&\approx0.00995050
	<0.01\\
	\delta_{\text{High}}(g=2,\,D=10^{-4})
	&\approx0.00999947
	<0.01
	\end{align*}
	It can also be shown that for a fixed $g$ that $\delta_\text{Low}$ and $\delta_\text{High}$ are monotonically increasing and decreasing in $D$ respectively, and that for a fixed $D$ both decrease monotonically with $g$. As such we can conclude that $\delta\leq 0.01$ for all $D\in\left[0,1\right]$ and $g\geq 2$. 
\end{proof}

\begin{clmhand}[\ref{clm:step2}]
	\label{clm:step2proof}
	The set resulting from \textsc{Trim}, $S_i^{(2)}$, is ($i$,\,$p_1$,\,$dsb+q^2$,\,$\delta=1/100$)-viable, where 
	\[ p_1(n),\,q(n)=n^{a}\quad\text{where }a=\mathcal{O}(1+\epsilon^{-1}\sqrt{\log \eta^{-1}/\log n})\,. \]
\end{clmhand}
\begin{proof}
	As shown above, marginalising over distinguishability for degeneracies at least two, the overlap error is upper bounded by $1/100$. The cardinality contribution comes in two parts: the results given by Program \ref{prog:trim}, and the previously obtained left-Schmidt vectors that are recycled. Let the cardinality of the Program~\ref{prog:trim} results be $r$ where
	\[ r=\underbrace{g10^9}_{\text{E-values}}
	\times \underbrace{\mathcal{O}(1/\epsilon)}_{\text{Energy net}}
	\times \underbrace{\vphantom{\mathcal{O}(1/\epsilon)}2^{{\mathcal{O}}(g/\epsilon)^{{\mathcal{O}}(1/\epsilon)}}}_{\text{BC net}}
	\times \underbrace{\vphantom{\mathcal{O}(1/\epsilon)}2^{\mathcal{O}(\epsilon^{-1}+\epsilon^{-1/4}\log^{3/4}(g/\epsilon))}}_{\text{Schmidt rank}}
	=2^{g^{\tilde{\mathcal{O}}(1/\epsilon)}}\,. \]
	The cardinality contribution from recycled Schmidt vectors is bounded by the bond dimension of each previous ground state approximation, as given in \cref{subsubsec:nondegen}. This contribution is $q(n)=n^a$ where $a=\mathcal{O}(1+\epsilon^{-1}\sqrt{\log \eta^{-1}/\log n})$, and dominates over $r$. The overall cardinality is the summation of these two contributions, and as such is given by $p_1(n)=n^a$ where $a$ has the scaling given above.
	
	The results of Program \ref{prog:trim} are linear combinations of vectors in $\Span(S_i^{(1)}\cup L_{h,i})$. The contributions from both spaces to the bond dimension of the worst-case result is the product of their cardinalities and bond dimensions respectively. As such these contributions are $dsb$ for $S_i^{(1)}$, and $q^2$ for $L_{h,i}$ respectively. 
\end{proof}

\subsection{Bond truncation}
\label{subsec:truncation analysis}

The idea here is to use the existence of low-rank approximate ground states to argue that truncation induces little error due to its optimal nature as a low-rank approximation~\cite{EckartYoung1936}. First we note an important result~\cite{LandauVaziraniVidick2013,Huang2014'} about truncation, that large overlap with a low-rank vector is maintained under truncation.

\begin{lem}[\!\!\cite{LandauVaziraniVidick2013,Huang2014'}]
	\label{lem:schmidt overlap}
	Given a vector $\ket{v}$ with Schmidt rank $R$ across $i|i+1$, for any $\ket{u}$ 
	\begin{align*}
	\abs{\braket{\Trunc^i_{R/\delta}(u)}{v}}\geq\abs{\braket{u}{v}}-\delta\,.
	\end{align*}
\end{lem}
Let $\ket{u}$ denote a witness of $S_{i}^{(2)}$ such that there exists a ground state $\ket{\Gamma}$ such that 
\[ \abs{\braket{\Gamma}{u}}\geq 1-\frac{1}{100}\,. \]
Fix $\ket{\Gamma}$ to denote the ground state with which $\ket{u}$ has maximal overlap. Next we invoke \cref{lem:area law}, taking  $r(n):=B_{1/800n}$ such that there exists an MPS $\ket{\gamma}$ with bond dimension bounded by $r(n)$ such that
\[ \abs{\braket{\Gamma}{\gamma}}\geq 1-\frac{1}{800}\,. \]
Applying \cref{lem:omnibus} we can combine these, showing $\ket{\gamma}$ and $\ket{u}$ have large overlap
\[ \abs{\braket{\gamma}{u}}\geq 1-2\left(\frac{1}{800}+\frac{1}{100}\right)= 1-\frac{18}{800}\,.\]
Next take $\ket{w}$ to be the renormalised truncation of $\ket{u}$ across every bond
\[ \ket{w}:=\frac{\Trunc^1_{P}\dots\Trunc^{n-1}_{P}\ket{u}}{\norm{\Trunc^1_{P}\dots\Trunc^{n-1}_{P}\ket{u}}}\,, \]
where we have taken $P(n):=800n\,r(n)$.

Applying \cref{lem:schmidt overlap} across every bond we get that the overlap-error induced by truncation grows additively, giving a final overlap of
\[ \abs{\braket{\gamma}{w}}\geq \abs{\braket{\gamma}{u}}-n/(800n)\geq 1-\frac{19}{800}\,. \]
Once again applying \cref{lem:omnibus} we get therefore that the ground space overlap of $\ket{w}$, which is lower bounded by the overlap with $\ket{\Gamma}$, can in turn be lower bounded
\[ \abs{\braket{\Gamma}{w}}\geq 1-2\left(\frac{1}{800}+\frac{19}{800}\right) =1-\frac{1}{20} \,.\]

Once again we must turn our attention now to orthogonalising $\ket{w}$ with respect to the existing ground states to satisfy the relevant orthogonality condition. The procedure for bounding this once again is to note that because $\ket{u}$ is orthogonal to $\Span\lbrace\ket{\gamma_1},\ldots,\ket{\gamma_{h-1}}\rbrace$, any overlap $\ket{w}$ has with this space is upper bounded by the component of $\ket{w}$ orthogonal to $\ket{u}$. Applying \cref{lem:omnibus} once again to their mutual overlap with $\ket{\Gamma}$, we get that 
\[ \abs{\braket{w}{u}}\geq 1-2\left(\frac{1}{100}+\frac{1}{20}\right) = 1-\frac{6}{50} \]
and thus that the orthogonal component is bounded \[ \sqrt{1-\abs{\braket{u}{w}}^2}\leq \sqrt{1-44^2/50^2}\leq 1/2\,. \]
By \cref{lem:ortho} then we can orthogonalise $\ket{w}$ with respect to the existing ground state approximations, at most tripling the overlap-error. Thus, after orthogonalisation, there must exist a ground state $\ket{\Gamma'}$
\[ \abs{\braket{\Gamma'}{w}}\geq 1-\frac{3}{20}\geq 1-\frac{1}{5}\,. \]

\begin{clmhand}[\ref{clm:step3}]
	\label{clm:step3proof}
	The set resulting from \textsc{Truncate}, $S_i^{(3)}$, is ($i$,\,$p_1+q$,\,$p_2$,\,$\delta=1/5$)-viable, where 
	\[ p_2(n)=n^{a}\quad\text{where}\quad a=\mathcal{O}(1+\epsilon^{-1}\sqrt{\log \eta^{-1}/\log n})\,. \]
\end{clmhand}
\begin{proof}
	The error level $\delta=1/5$ is proven above. As the truncation procedure is performed element-wise, the cardinality is changed only by the reintroduction of recycled Schmidt vectors, adding $q$.
	
	After truncation, the bond dimension is bounded by the trimming level
	\[ P=800nB_{1/800n}=n^{1+o(1)}\,. \]
	The orthogonalisation procedure then will additively increase this bond dimension by $q$, the bond dimension of the previous approximate ground states. As such the final bond dimension can be bounded by a polynomial $p_2$ of the same scaling.
\end{proof}

\subsection{Error reduction}
\label{subsec:reduction analysis}

Next we need to address the \textsc{Reduce} and \textsc{FinalReduce} procedures. We will do this by considering the construction of an AGSP $A$, and an approximation thereof $K$. Both operators are subject to a free parameter $\zeta$, which controls the magnitude of the error reduction they can perform. The only difference between \textsc{Reduce} and \textsc{FinalReduce} will be the final level of $\zeta$ utilised. As such we will analyse the procedure for an arbitrary $\zeta$, stating the \textsc{Reduce} and \textsc{FinalReduce} as special cases of a more general analysis.

\subsubsection{Exact AGSP}
First we consider the exact AGSP construction
\[A:=\exp\left[-\frac{x(H-\epsilon_0')^2}{2\epsilon^2}\right]\,,\]
where $\epsilon_0':=\braopket{\gamma_1}{H}{\gamma_1}$ is an inverse polynomial approximation to the ground state energy, and $x$ is a $\zeta$-dependent parameter given in \cref{app:aagsp}. As well as specifying the precise parameters of the AGSP/AAGSP construction, in \cref{app:aagsp} we also show that for each ground state $\ket{\Gamma}$ and excited state $\ket{\Gamma^\perp}$:
\begin{align*}
\norm{A\ket{\Gamma}}&\geq\frac{19}{20} &  
\norm{A\ket{\Gamma^\perp}}&\leq \frac{\zeta}{2}
\end{align*}
It is in this sense that our AGSP approximates a ground state projector. 

Take $\ket{u}$ to be a witness of $S_{i}^{(3)}$, such that there exists a ground state $\ket{\Gamma}$ with
\[ \abs{\braket{\Gamma}{u}}\geq 1-1/5\,. \]
We can decompose this witness in the energy eigenbasis \[\ket{u}=\sum_{j= 0}^{d^n-1}\sqrt{p_j}\ket{E_j}\,,\] where $\ket{E_0}$ is a ground state, $\lbrace\ket{E_j}\rbrace_{j> 0}$ are all excited states\footnote{As we only require that the vectors $\ket{E_j}$ be energy eigenstates up to rephasing, all vectors appearing in this decomposition from the same energy-eigenspace can be merged into a single vector, allowing for this decomposition to be considered non-degenerate without loss of generality.} with energy $\epsilon_j\geq \epsilon_0+\epsilon$. Due to the role of $\ket{u}$ as a witness of $S_i^{(3)}$ we have $p_0\geq 16/25$, and due to normalisation $\sum_{j}p_j=1$. Lastly take $\ket{w}$ and to be the result of applying the AGSP to the witness and renormalising
\[ \ket{w}:=\frac{A\ket{u}}{\norm{A\ket{u}}}\,. \]

We can bound the renormalising denominator term by applying $A$ to $\ket{u}$ in this decomposed form
\begin{align*}
\norm{A\ket{u}}&=\norm{\sum_{j=0}^{d^n-1}\sqrt{p_j}A\ket{E_j}}\\
&\geq \sqrt{p_0}\norm{A\ket{E_0}}-\norm{\sum_{j\geq1}\sqrt{p_j}A\ket{E_j}}\\
&\geq\frac{4}{5}\cdot\frac{19}{20}-\frac{3}{10}\zeta\,.
\end{align*}
If we assume $\zeta\leq 1$ then $\norm{A\ket{u}}\geq 1/5$.
Using this bound we can argue that because $A$ shrinks the high energy terms, yet does not shrink the magnitude all that much, the energy of $\ket{w}$ should be a drastic improvement over that of $\ket{u}$.

\begin{lem}
	\[ \braopket{w}{H}{w}\leq \epsilon_0+\zeta\epsilon/\num{80000} \]
\end{lem}
\begin{proof}
	We can bound the energy by simply expanding the expectation of $H':=(H-\epsilon_0)$, and seeing how it changes under the application of $A$ using the above decomposition in the energy basis.
	\begin{align*}
	\braopket{w}{H'}{w}
	&=\frac{\braopket{u}{AH'A}{u}}{\norm{A\ket{u}}^2}\\
	&\leq 25\sum_{j,k\geq 0}p_j\braopket{E_j}{AH'A}{E_k}\,.\\
\intertext{Because $A$ commutes with $H$, the double sum collapses and we have}
	&\leq 25\sum_{j\geq 0}p_j\braopket{E_j}{AH'A}{E_j}\\
	&\leq 25\sum_{j\geq1}p_j(\epsilon_j-\epsilon_0)\exp(-x(\epsilon_j-\epsilon_0')^2/\epsilon)\,.\\
\intertext{Using $\epsilon_0' - \epsilon_0 \le \epsilon/2$, and then using the H\"{o}lder inequality, we have}
	&\leq 50\sum_{j\geq1}p_j(\epsilon_j-\epsilon_0')\exp(-x(\epsilon_j-\epsilon_0')^2/\epsilon)\\
	&\leq 50\left[\max_{E\geq \epsilon/2}E\exp(-xE^2/\epsilon^2)\right]\left[\sum_{j\geq1}{p_j}\right]\,.\\
	\intertext{The maximum is on the boundary because $x$ is sufficiently large, so we find}
	&\leq 25 \epsilon\, \mathrm{e}^{-x/4}
	\leq \zeta\epsilon/\num{3200}
	\end{align*}
	Where above we have used the chosen value of $x:=33-8\log(\zeta)$ from \cref{app:aagsp}.
\end{proof}

\subsubsection{Approximate AGSP}
\label{subsubsec:AAGSP}

As well as the bounds on the power of $A$, \cref{app:aagsp} also gives that we can approximate $A$ by a low bond dimension operator $K$ with
\begin{align}\label{eq:zetaprime}
	\norm{A-K}\leq\zeta':=\frac{\zeta\epsilon}{\num{240000}n}\,.
\end{align}
For simplicity we can express $K=A+\zeta'O$, where $\norm{O}\leq 1$. As with $K$ we can bound the renormalising term for $\ket{u}$,
\begin{align*}
\norm{K\ket{u}}
&=\norm{(A+\zeta'O)\ket{u}}\\
&\geq \norm{A\ket{u}}-\zeta'\norm{O\ket{u}}\\
&\geq \frac{1}{5}-\zeta'\,.
\end{align*}
Now if we take $\zeta'\leq 1/10$ then we get $\norm{K\ket{u}}\geq 1/10$. As such we can replace the state $\ket{w}$ that had an exact AGSP with an approximate version using the AAGSP,
\[ \ket{w'}:=\frac{K\ket{u}}{\norm{K\ket{u}}} \]
and similarly bound the energy.
\begin{lem}
	\[\braopket{w'}{H}{w'}\leq \epsilon_0+\zeta\epsilon/400\]
\end{lem}
\begin{proof}
	Once again we consider the expectation value of $H'$
	\begin{align*}
	\braopket{w'}{H'}{w'}
	=&\braopket{w'}{H'}{w'}\\
	=&\frac{\braopket{u}{KH'K}{u}}{\norm{K\ket{u}}^2}\\
	\leq&100\braopket{u}{(A+\zeta'O)H'(A+\zeta'O)}{u}\\
	\leq&100\braopket{u}{AH'A}{u}+100\zeta'^2\braopket{u}{OH'O}{u}\\
	&+100\zeta'\braopket{u}{OH'A}{u}+100\zeta'\braopket{u}{AH'O}{u}\\
	\leq &4\braopket{w}{H}{w}+300n\zeta'
	\\\leq&
	\zeta\epsilon/800+
	\zeta\epsilon/800=
	\zeta\epsilon/400			
	\end{align*}
	Reintroducing $\epsilon_0$ gives the desired energy bound.
\end{proof}
Finally we must orthogonalise to negate any overlap with the existing vectors this has caused. 
\begin{lem}
For any $\ket{\gamma} \in \Span\{\ket{\gamma_1},\ldots,\ket{\gamma_h-1}\}$, $\abs{\braket{w'}{\gamma}}\le \frac{400-1}{400+1}$.
\end{lem}

\begin{proof}
The original witness $\ket{u}$ has a large ground state overlap, which in terms of the previous decomposition means 
\[ \sqrt{p_0}=\braket{E_0}{u}\geq 4/5\,. \]
Using the fact that $\norm{K\ket{u}}\leq 1$ and that $A$ is diagonal in the energy basis, we can bound the overlap that $\ket{w'}$ in turn has with $\ket{E_0}$. Using $\zeta' \le 1/100$, 
\begin{align*}
\abs{\braket{E_0}{w'}}
&=\frac{\abs{\braopket{E_0}{K}{u}}}{\norm{K\ket{u}}}\\
&\geq \abs{\braopket{E_0}{K}{u}}\\
&\geq \biggl\vert\sum_j\sqrt{p_j}\braopket{E_0}{A}{E_j}\biggr\vert-\zeta'\\
&\geq \sqrt{p_0}\braopket{E_0}{A}{E_0}-\zeta'\\
&\geq \frac{4}{5}\frac{19}{20}
-\frac{1}{100}
= \frac{3}{4}\,.
\end{align*}
Applying \cref{lem:omnibus} we therefore get that 
\[ \abs{\braket{u}{w'}}\geq 1-2\left(\frac{1}{5}+\frac{1}{4}\right)=\frac{1}{10}\,. \]
As has been argued before, the component of $\ket{w'}$ overlapping with $\Span\lbrace\ket{\gamma_1},\dots,\ket{\gamma_{h-1}}\rbrace$ can be bounded by the component perpendicular to $\ket{u}$. This component is bounded
\[ \sqrt{1-\abs{\braket{u}{w'}}^2}=\frac{3\sqrt{11}}{10}\leq \frac{400-1}{400+1}\,. \]
\end{proof}

By \cref{lem:ortho} and the previous lemma, we can thus orthogonalise $\ket{w'}$ to a vector $\ket{w''}$ with a $400$-fold increase in the excitation, resulting in an energy of 
\[ \braopket{w''}{H}{w''}\leq\epsilon_0+\zeta\epsilon\,. \]
As such setting $\zeta$ to the desired final energy error level and applying the local Schmidt operators of $K$, the energy error can be reduced to $\zeta$.

\begin{clmhand}[\ref{clm:step4},\ref{clm:step4final}]
	\label{clm:step4proof}
	
	The set resulting from \textsc{Reduce}, $S_{i}^{(4)}$, is ($i$,~$pp_1+pq+q$,~$pp_2$,~$\Delta=c\epsilon^6/f^4$)-viable, where 
	\[ p(n)=n^{\mathcal{O}(1)}\,. \]
	The set resulting from \textsc{FinalReduce}, $S_n^{(4)}$, is ($n$,~$p_1+q+g$,~$p_0p_2$,~$\Delta=\eta^2/4f$)-viable, where 
	\[ p_0(n)=n^{a}\quad\text{where}\quad a=\mathcal{O}(1+\epsilon^{-1}\sqrt{\log \eta^{-1}/\log n})\,. \]
\end{clmhand}
\begin{proof}
The construction described above gives a $K$ with a bond dimension that scales (see Ref.~\cite{Osborne2006}) as $n^{a'}$ where $a'=\mathcal{O}(1+\epsilon^{-1}\sqrt{\log \zeta^{-1}/\log n})$.
	
	For \textsc{Reduce} the cardinality growth is caused by the fact that $K$ must be applied locally, i.e.\ each Schmidt operator must be individually applied. The number of these operators is bounded by the bond dimension of $K$. As such \textsc{Reduce}, for which $\zeta=c\epsilon^6/f^4=\mathcal{O}(1)$, has a multiplicative growth in cardinality of $p(n)=n^{\mathcal{O}(1)}$. On top of this the recycled Schmidt vectors additively increase the cardinality once more by $q$.
	
	For \textsc{FinalReduce} however the whole operator can be applied without Schmidt decomposition, thus leaving cardinality unchanged. The recycled previous ground states grow the cardinality additively by $g$.
	
	The bond dimension increase is once again multiplicative, and once again bounded by the bond dimension of $K$ as an MPO. For \textsc{Reduce} this means a growth by a factor of $p(n)$. For \textsc{FinalReduce} however $\zeta=\eta^2/4f$ is not necessarily constant, as such the growth in bond dimension has the more general scaling of 
	\[ p_0(n)=n^{a}\quad\text{where}\quad a=\mathcal{O}(1+\epsilon^{-1}\sqrt{\log \eta^{-1}/\log n})\,. \]
\end{proof}

Here we briefly comment on the improved analysis of Ref.~\cite{Huang2014'''}. Suppose we perform the following for our \textsc{FinalReduce}: instead applying of a single powerful AGSP, we interleave a mild AGSP with bond trimming multiple times. With sufficient parameters (the details are given in Sec.5.5 of \cite{Huang2014'''}), this can be done such that the bond dimension remains at most $n^{\mathcal{O}(1)}$, and the error level is halved. Performing this procedure the requisite number of times, the scaling of final bond dimension can be improved to $p_0=n^{\mathcal{O}(1)}$ for $\eta^{-1}=n^{\mathcal{O}(1)}$. Carrying this analysis through, this also reduces the scaling of $p_1$ and $q$ to both be $n^{\mathcal{O}(1)}$, and therefore a final run-time also of $T=n^{\mathcal{O}(1)}$.

\section{Conclusion}\label{sec:conclusion}

We have shown the ground state approximation of Refs.~\cite{LandauVaziraniVidick2013,Huang2014'} can be extended to degenerate systems with an approximation of the ground space projector that is up to inverse polynomial error in any Schatten norm. 

AGSPs have proven powerful tools for relating the structural properties of gapped systems to various computational complexity bounds regarding them. Most AGSP constructions are either functions of the Hamiltonians~\cite{AradKitaevLandauVazirani2013,Huang2014,Huang2014'} or parts of the Hamiltonian terms~\cite{AharonovAradLandauVazirani2008,AharonovAradLandauVazirani2010,AradLandauVazirani2012,AharonovAradLandauVazirani2011,LandauVaziraniVidick2013}. While we utilise such AGSPs in intermediate steps, our algorithm outputs a state-based AGSP, and more thoroughly exploits the underlying one-dimensional structure, allowing for a more powerful overall construction. While previous AGSP constructions have typically been full-rank, our AGSP takes the form of a minimal-rank projector, that moreover has a polynomially bounded bond dimension as a matrix product operator (MPO). It remains an open question for which systems there exist efficiently constructible AGSPs that are also \emph{low-rank} AGSPs. Note that this is important for applications because while AGSPs that are close to the ground space in operator norm are easy to construct even with small bond dimension in MPO form, they do not generally allow for low-error estimation of expectation values of other MPO observables. 

While running in polynomial-time, the specific runtime scaling of the above algorithm is highly suboptimal, leaving wide room for further optimisation. One potential avenue for optimisation is the level of redundancy within the viable sets that are generated. While there are rigorous upper bounds on the cardinality of these sets, no care has been taken to explore to possibility of high levels of linear dependence, opening the door to far better practical runtimes than our analysis would suggest. 

\emph{Note added in preparation:} This work began as the undergraduate Honours thesis of the first author as a project to extend LVV to the degenerate case. During that time, the results of Huang in Ref.~\cite{Huang2014'} (version 1) appeared, and were incorporated into our independent analysis. Near the completion of this manuscript, Huang also extended the results in Ref.~\cite{Huang2014'} to the degenerate case and published this as Ref.~\cite{Huang2014''} (version 3), an updated version of the original manuscript. While our techniques rely on the methods in version 1 (Ref.~\cite{Huang2014'}), our treatment of the degenerate case was made completely independently from version 3 (Ref.~\cite{Huang2014''}). 

Our results and the new results in Ref.~\cite{Huang2014''} are essentially the same, but with the following differences. The most notable is that the algorithm of Ref.~\cite{Huang2014''} is single-pass, finding a full set of ground states with a single sweep of the system. In principle our multi-pass algorithm will generate viable sets with far higher levels of redundancy, causing worse scaling in the parameters of \cref{tab:params}. In the final runtime however this slow-down only influences the various constants hidden under big-O terms, giving an overall scaling that is identical to that of Ref.~\cite{Huang2014''}. The second major difference regards the size-trimming, the most complicated step to generalise. Our algorithm works by applying an energy constraint and minimising the overlap between witnesses, while Ref.~\cite{Huang2014''} constrains the overlap and minimises the energy. This allows one to use boundary contraction to circumvent the issue of local distinguishability, simplifying the analysis of this step. The final difference is that our algorithm does not require knowing a specific bound on the ground state degeneracy $g$, only that a bound of the form $g = \mathcal{O}(1)$ exists, whereas it is not clear how to perform the algorithm in Ref.~\cite{Huang2014''} without an explicit upper bound. 

Subsequent to the journal submission of this manuscript, but prior to its publication, Arad et.al.\ released a yet further improved algorithm in Ref.~\cite{AradLandauVaziraniVidick2016}. While our algorithm works by iterating along the spin chain, constructing matrix product states, the main new technique in Ref.~\cite{AradLandauVaziraniVidick2016} is to instead apply divide-and-conquer approach, generating states in a tree-like fashion. This method allows for much tighter runtime analysis (e.g.\ runtime reduction from $\poly(n)$ to $\mathcal{O}(n^{3.38})$ for frustration-free systems), as well as extensions to much larger classes of systems (e.g.\ a quasi-polynomial-time algorithm for certain gapless models).

We thank Thomas Vidick for comments on the previous version of this manuscript. We acknowledge support from the Australian Research Council via EQuS project number CE11001013, the US Army Research Office via grant numbers W911NF-14-1-0098 and W911NF-14-1-0103, and by iARPA via the MQCO program. S.T.F.\ also acknowledges support from an ARC Future Fellowship FT130101744.


\appendix


\section{Overlap lemma}
\label{app:overlap}

In this appendix we want to relax the notion of a basis, and show that key properties thereof are robust to a small amount of allowed error. 

Suppose we have some $g$-dimensional space, represented by a $g$-rank projector $G$ onto this space. Take $\ket{v_i}$ for $1\leq i\leq g$ to be an orthonormal set of vectors, each with $\braopket{v_i}{G}{v_i}\geq 1-\delta/g$. First we will prove that all members of $S:=\Span\lbrace\ket{v_i}\rbrace$ have non-trivial overlap with $G$. We will then use this to show that the vectors $G\ket{v_i}$ span the space. Finally we will show that any vector orthogonal to $\Span\lbrace\ket{v_i}\rbrace$ has low overlap with the space.
\begin{applem}
	\label{applem:overlap}
	Any vector $\ket{v}\in S$ has $\braopket{v}{G}{v}\geq1-\delta$.
\end{applem}
\begin{proof}
	%
	As $\ket{v}$ lies in $S$, we can write it as a linear combination of the form
	\[\ket{v}:=\sum_{i=1}^{g}c_i\ket{v_i}\,.\]
	Next we can bound the expectation value of the complementary projector $I-G$ by using the triangle and Cauchy-Schwartz inequalities. 
	\begin{align*}
	\braopket{v}{(I-G)}{v}
	&=\sum_{i,j}c_i^*c_j\braopket{v_i}{(I-G)}{v_j}\\
	&\leq\sum_{i,j}\abs{c_i}\abs{c_j}\abs{\braopket{v_i}{(I-G)}{v_j}}\\
	&\leq\left[\sum_{i=1}^g\abs{c_i}\norm{(I-G)\ket{v_i}}\right]^2\\
	&\leq \delta
	\end{align*}
	where we have used the normalisation $\sum_{i}\abs{c_i}^2=1$ and the resulting bound $\sum_{i}\abs{c_i}\leq \sqrt{g}$. This in turn gives $\braopket{v}{G}{v}\geq1-\delta$.
\end{proof}

\begin{applem}[Basis under projection]
	\label{applem:basis}
	The projections $G\ket{v_i}$ form a basis of $\mathrm{im}(G)$. 
\end{applem}
\begin{proof}
	As the number of vectors matches the dimension of the target space, proving linear independence is sufficient to prove the formation of a basis. If the projected vectors were linearly dependent then this would imply there exist a non-zero $\ket{v}\in S$ such that $G\ket{v}=0$. \cref{applem:overlap} however gives $\braopket{v}{G}{v}\geq 1-\delta\geq 0$ for all such vectors, thus proving linear independence by contradiction.
\end{proof}

Next we want to generalise the concept of a \emph{full} basis, specifically the fact that any vector perpendicular to a full basis is necessarily perpendicular to the entire space, meaning that there is a natural maximum to the size of any given basis.

\begin{applem}[Fullness]
	\label{applem:fullness}
	Any vector $\ket{v'}$ which is orthogonal to $S$ has $\braopket{v'}{G}{v'}\leq\delta$.
\end{applem}
\begin{proof}
	As the projection of this space is spanning, we can pick $\ket{v}\in S$ such that it projects into the same vector as $\ket{v'}$, i.e. $G\ket{v}\propto G\ket{v'}$. As such we can decompose both states as
	\begin{align*}
	\ket{v}&=\sqrt{\lambda}\ket{g}+\sqrt{1-\lambda}\ket{h}\\
	\ket{v'}&=\sqrt{\lambda'}\ket{g}+\sqrt{1-\lambda'}\ket{h'}
	\end{align*}
	where $\ket{g}\in\mathrm{im}(G)$ and $\ket{h},\ket{h'}\in\mathrm{ker}(G)$. The orthogonality of $\ket{v}$ and $\ket{v'}$ gives 
	\[ 0=\braket{v}{v'}=\sqrt{\lambda}\sqrt{\lambda'}+\sqrt{1-\lambda}\sqrt{1-\lambda'}\braket{h}{h'}\,. \]
	Using $\abs{\braket{h}{h'}}\leq 1$ in turn gives $\lambda+\lambda'\leq 1$. \cref{applem:overlap} gives that $\lambda\geq 1-\delta$, leading us to conclude that $\lambda'\leq \delta$, and thus that $\braopket{v'}{G}{v'}\leq \delta$ generically for any vector orthogonal to $S$.
\end{proof}


\section{Approximate AGSP construction}
\label{app:aagsp}
In this section we explicitly lay out the various errors associated with approximating the AGSP defined in Sections~\ref{subsec:reduction} and \ref{subsec:reduction analysis}. The scaling of these errors was stated in less detail in Ref.~\cite{Huang2014'}. First we consider the AGSP itself
\begin{align*}
A:=\exp\left[-\frac{x(H-\epsilon_0')^2}{2\epsilon^2}\right]=\frac{\epsilon}{\sqrt{2\pi x}}\int_{-\infty}^{+\infty}\exp\left[-\frac{\epsilon^2t^2}{2x}-i(H-\epsilon_0')t\right]\dif t
\end{align*}
where $\epsilon_0'$ is the ground energy approximation given by $\ket{\gamma_1}$. If we take $\eta,\zeta \leq 1/2$ and assume that $\zeta\geq \eta^2/4f$ then
\[\abs{\epsilon_0'-\epsilon_0}\leq \frac{\eta^2\epsilon}{4f}\leq\epsilon\zeta
\leq\frac{\epsilon}{2}\,.\]
The magnitude of this AGSP on the excited subspace, known as the \emph{shrinking factor}, can be bounded 
\[\norm{A\ket{\Gamma^\perp}}\leq \exp\left[-\frac{x(\epsilon/2)^2}{2\epsilon^2}\right]=e^{-x/8}\,.\]
If we define ${x:=33-8\log \zeta}$, then this shrinking factor is at most $\zeta/2$. Further assuming that $\zeta \le 1/25$, the magnitude of $A$ on the ground space can be lower bounded 
\begin{align*}
\norm{A\ket{\Gamma}}
&=\exp\left[-\frac{x(\epsilon_0-\epsilon_0')^2}{2\epsilon^2}\right]\\
&\geq \frac{19}{20}\,.
\end{align*}

Next we move on to the approximate AGSP, defined by
\begin{align*}
K:=\frac{2\epsilon}{\sqrt{2\pi x}}\sum_{j=0}^{\lceil T/\tau\rceil}\exp\left[i\epsilon_0'\tau j-\frac{\epsilon^2\tau^2j^2}{2x}\right]U_D(\tau j)\,.
\end{align*}
where $T$, $\tau$ and $D$ are parameters and $U_D$ is some unitary operator all of which we shall define below. This definition relies on taking the Fourier decomposition of $A$, and approximating this by truncating and discretising the integral and approximating the propagators by a known algorithm~\cite{Osborne2006}.

\subsection{Time Truncation Error}
The truncation error is due to the finite range of the integral, controlled by the parameter $T$. If we let $\delta_T$ denote the truncation error (in the operator norm) then it can be bounded
\begin{align*}
\delta_T&=2\norm{\frac{\epsilon}{\sqrt{2\pi x}}\int_{T}^{\infty}\exp\left[-\frac{\epsilon^2t^2}{2x}-iHt\right]\dif t}\\
&\leq\frac{2\epsilon}{\sqrt{2\pi x}}\int_{T}^{\infty}\exp\left[-\frac{\epsilon^2t^2}{2x}\right]\dif t\\
&=\frac{\epsilon}{\sqrt{2x}}\text{erfc}\left(\frac{\epsilon T}{\sqrt{2x}}\right)\\
&\leq \exp\left[-\frac{\epsilon^2T^2}{2x}\right]\,,
\end{align*}
where the last inequality holds for sufficiently large $T$, e.g.\ $T \ge 1$. For the time-scale
\begin{align*}
T:= 
\frac{\sqrt{2x}}{\epsilon}\sqrt{\log3/\zeta'}=\mathcal{O}\bigl(\epsilon^{-1}\log(1/\zeta)\sqrt{\log (n/\zeta)}\bigr)
\end{align*} 
we can in turn upper bound $\delta_T$ by $\delta_T \leq \zeta'/3=\Omega(\zeta/n)$, where $\zeta'$ is defined in \cref{eq:zetaprime}.

\subsection{Discretisation Error}

Next we consider the discretisation error associated with approximating this integral with a Riemann sum. The term being approximately integrated is a Gaussian centred on the origin and so its absolute value is even as a function of $t$, allowing us to consider only the $t>0$ error and doubling it. On this range the absolute value is not only even, but monotonically decreasing. As such the concepts of left and right Riemann sums correspond to upper and lower Riemann sums. If we let the discretisation error be $\delta_D$ then
\begin{align*}
\delta_D&=2\frac{\epsilon}{\sqrt{2\pi x}}\norm{\sum_{j=0}^{\lceil T/\tau\rceil }\tau\exp\left[-\frac{\epsilon^2\tau^2j^2}{2x}-iH\tau j\right]-
	\int_{0}^{T}
	\exp\left[-\frac{\epsilon^2t^2}{2x}-iHt\right]\dif t
}\,.
\intertext{This error can be upper bounded in turn by the difference between upper and lower Riemann sums which is equal to the left and right Riemann sums owing to the monotonicity of the integrand.}
\delta_D&\leq2\frac{\epsilon\tau}{\sqrt{2\pi x}}\norm{\sum_{j=0}^{\lceil T/\tau\rceil}\exp\left[-\frac{\epsilon^2\tau^2j^2}{2x}-iH\tau j\right]-
	\sum_{j=1}^{\lceil T/\tau\rceil+1}\exp\left[-\frac{\epsilon^2\tau^2j^2}{2x}-iH\tau j\right]
}\,.
\end{align*}
The difference between these two summations is telescopic. As such it equals the difference between the highest and lowest terms.
\begin{align*}
\delta_D&\leq2\frac{\epsilon\tau}{\sqrt{2\pi x}}\norm{1-\exp\left[-\frac{\epsilon^2T^2}{2x}-iHT\right]}\,.
\end{align*}
Using the trivial bound on the real exponential term, we get
\begin{align*}
\delta_D&\leq2\frac{\epsilon\tau}{\sqrt{2\pi x}}\norm{1-\exp\big[-iHT\big]\Big.}\\
&\leq4\frac{\epsilon\tau}{\sqrt{2\pi x}}\,.
\end{align*}
Again this error can also be bounded by $\delta_D\leq \zeta'/3$ if we take the time-step
\begin{align*}
\tau^{-1}:= \frac{12}{\zeta'}\sqrt{\frac{\log(3/\zeta')}{\pi}}=
\frac{12\epsilon}{\zeta'\sqrt{2\pi x}}= 
\mathcal{O}\left(\frac{n}{\zeta \sqrt{\log 1/\zeta}}\right)\,.
\end{align*}

\subsection{Bond Truncation Error}

To upper bound the total truncation error for the entire approximation by $\zeta'/3$, it is sufficient to bound the truncation error for each propagator $\exp(-iH\tau j)$ by $\zeta'/3$ since the total sum is normalized. Using the construction of Ref.~\cite{Osborne2006} this error can be achieved if the bond dimension of each propagator is
\begin{align*}
D&=\exp{\left[\mathcal{O}\left(\tau j\right)+\mathcal{O}\left(\log(n^2/\zeta')\right)\right]}\\
&\leq\exp\left[\mathcal{O}\left(T\right)+\mathcal{O}\left(\log n/\zeta\right)\right]\\
&=2^{\mathcal{O}(T)}\poly(n/\zeta)\\
&=2^{\mathcal{O}\bigl(\epsilon^{-1}\log(1/\zeta)\sqrt{\log(n/\zeta)}\bigr)}\poly(n/\zeta)\,.
\end{align*}
As bond dimension is additive, and there are $2T/(\tau+1)$ terms, the total bond dimension for the AGSP $B$ is
\begin{align*}
B&\leq2 T\tau^{-1}D\\
&=\mathcal{O}\bigl(\epsilon^{-1}\log(1/\zeta)\sqrt{\log (n/\zeta)}\bigr)
\mathcal{O}\left(\frac{n}{\zeta \sqrt{\log 1/\zeta}}\right)
2^{\mathcal{O}\bigl(\epsilon^{-1}\log(1/\zeta)\sqrt{\log(n/\zeta)}\bigr)}\poly(n/\zeta)\\
&=
2^{\mathcal{O}\bigl(\epsilon^{-1}\log(1/\zeta)\sqrt{\log(n/\zeta)}\bigr)}\poly(n/\zeta)\\
&=n^{\mathcal{O}(1+\epsilon^{-1}\sqrt{\log(1/\zeta)/\log(n)})}
\end{align*}
where we have assumed $\zeta^{-1}=n^{\mathcal{O}(1)}$.


\section{Frustration bound}
\label{app:frust}

Consider a 1D system with nearest neighbour interactions, were each local Hamiltonian term $H_i$ is bounded $0\leq H_i\leq \one$. For a cut $i|i+1$ we split the Hamiltonian into a left-Hamiltonian $H_L=\sum_{j=1}^{i-1}H_j$, middle-Hamiltonian $H_M=H_i$, and right-Hamiltonian $H_R=\sum_{j=i+1}^{n}H_j$; we refer to their expectation values as the left/middle/right-energy.
\begin{applem}
	For a Hamiltonian normalised as above, the left-energies relative to any cut for any two states with an energy less than $\Delta E$ over the ground, differ by no more than $1+\Delta E$.
\end{applem}
\begin{proof}
	Suppose we have two states $\ket{v_1}$ and $\ket{v_2}$ with local-energies given in \cref{apptab:energies}.
	\bgroup
	\renewcommand\tabcolsep{9pt}
	\begin{table}[th!]
		\centering
		\begin{tabular}{ccccc}
			\toprule
			State & $\left<H\right>$ & $\left<H_L\right>$ & $\left<H_M\right>$ & $\left<H_R\right>$\\ \midrule
			$\ket{v_1}$ & $E$ & $L$ & $E-L-R-\Delta R$ & $R+\Delta R$\\
			$\ket{v_2}$ & $E$ & $L+\Delta L$ & $E-L-R-\Delta L $ & $R$\\
			\bottomrule
		\end{tabular}
		\caption{Local energies of the two ground states across a bipartition.}
		\label{apptab:energies}
	\end{table}
	\egroup
	where $E=E_0+\Delta E$, and $E_0$ is the ground energy, and we take $\Delta L\geq0$ without loss of generality. Using the positivity of $H_M$ on the expectations for $\ket{v_2}$ gives
	\[L+R\leq E-\Delta L\,.\]
	As $\ket{v_1}$ has a left-energy $L$, it must possess at least one left Schmidt vector $\ket{l}$ with a left-energy at most $L$, similarly $\ket{v_2}$ has a right Schmidt vector $\ket{r}$ with right-energy at most $R$. As such we can construct a state 
	$\ket{v}=\ket{l}\otimes\ket{r}$
	with energy upper bounded
	\begin{align*}
	\braopket{v}{H}{v}
	&=\braopket{v}{H_L}{v}+\braopket{v}{H_M}{v}+\braopket{v}{H_R}{v}\\
	&=\braopket{l}{H_L}{l}+\braopket{v}{H_M}{v}+\braopket{r}{H_R}{r}\\
	&\leq L+1+R\\
	&\leq E+1-\Delta L\,.
	\end{align*}
	As $\braopket{v}{H}{v}\geq E_0$, we get that in general the difference between the left-energies is bounded
	\[\Delta L\leq 1+\Delta E\,.\]
A simple example of a perturbed Ising model shows that this bound is tight.
\end{proof}

\end{document}